\documentclass[10pt,
	a4paper,
	onecolumn,
	oneside,
	notitlepage,
	abstract=true%
]{scrartcl}

\usepackage[utf8]{inputenc}						
\usepackage[main=english,ngerman]{babel}			
\usepackage{lmodern}							
\usepackage[T1]{fontenc}							
\usepackage{amssymb,amsmath,amsthm,mathtools}	
\usepackage{amsfonts}							
\usepackage{lipsum}								
\usepackage{multirow}							
\usepackage{comment}							
\usepackage{booktabs}							
\usepackage{graphics,graphicx}					
\usepackage[dvipsnames]{xcolor}					
\usepackage{dsfont}								
\usepackage{enumitem}							
\usepackage[mathscr]{euscript}						
\usepackage{authblk}
\usepackage{threeparttable}

\DeclareUnicodeCharacter{00A0}{ }

\newcommand{\Fcal}{{\mathcal F}}
\newcommand{\FF}{{\mathbb F}}
\newcommand{\PP}{{\mathbb P}}

\newcommand{\na}{{\mathbb N}}
\newcommand{\re}{{\mathbb R}}
\newcommand{\replus}{{\re_+}}
\newcommand{\B}{{\mathcal B}}

\newcommand{\N}{{\mathcal N}}

\newcommand{\from}{{:}\penalty\binoppenalty\mskip\thickmuskip}
\newcommand{\indicatorset}[1]{\mathds{1}_{#1}}

\newcommand{\4}{\mathchoice{\mskip1.5mu}{\mskip1.5mu}{}{}}
\newcommand\sbullet[1][.5]{\mathbin{\vcenter{\hbox{\scalebox{#1}{$\bullet$}}}}}

\numberwithin{equation}{section}

\theoremstyle{plain}
\newtheorem{theorem}[equation]{Theorem}

\newtheorem{lemma}[equation]{Lemma}
\newtheorem{proposition}[equation]{Proposition}

\theoremstyle{definition}
\newtheorem{definition}[equation]{Definition}
\newtheorem{example}[equation]{Example}
\newtheorem{assumption}[equation]{General Assumption}
\newtheorem{convention}[equation]{General Convention}

\theoremstyle{remark}
\newtheorem{remark}[equation]{Remark}
\newtheorem*{notation}{Notation}

\DeclareMathOperator*{\argmin}{arg\,min}
\DeclareMathOperator{\esssup}{ess\,sup}
\DeclareMathOperator*{\trace}{tr}
\DeclareMathOperator*{\Cov}{Cov}

\makeatletter
\newcommand{\overbar}[1]{\skewbar{#1}{-1}{-1}{.25}}
\newcommand{\skewbar}[4]{{\overbarpalette\makeoverbar{#1}{#2}{#3}{#4}}#1}
\newcommand{\overbarpalette}[5]{\mathchoice
{#1\textfont\displaystyle{#2}1{#3}{#4}{#5}}
{#1\textfont\textstyle{#2}1{#3}{#4}{#5}}
{#1\scriptfont\scriptstyle{#2}{.7}{#3}{#4}{#5}}
{#1\scriptscriptfont\scriptscriptstyle{#2}{.5}{#3}{#4}{#5}}}
\newcommand{\makeoverbar}[7]{%
\setbox0=\hbox{$\m@th#2\mkern#5mu{{}#3{}}\mkern#6mu$}%
\setbox1=\null \dimen@=#4\fontdimen8#13 \dimen@=3.5\dimen@
\advance\dimen@ by \ht0 \dimen@=-#7\dimen@ \advance\dimen@ by \wd0
\ht1=\ht0 \dp1=\dp0 \wd1=\dimen@
\dimen@=\fontdimen8#13 \fontdimen8#13=#4\fontdimen8#13
\rlap{\hbox to \wd0{$\m@th\hss#2{\overline{\box1}}\mkern#5mu$}}
\fontdimen8#13=\dimen@}

\def\makeunderbar#1#2#3#4#5#6#7{%
\setbox0=\hbox{$\m@th#2\mkern#5mu{{}#3{}}\mkern#6mu$}%
\setbox1=\null \dimen@=#4\fontdimen8#13 \dimen@=3.5\dimen@
\advance\dimen@ by \dp0 \dimen@=-#7\dimen@ \advance\dimen@ by \wd0
\ht1=\ht0 \dp1=\dp0 \wd1=\dimen@
\dimen@=\fontdimen8#13 \fontdimen8#13=#4\fontdimen8#13
\rlap{\hbox to \wd0{$\m@th\hss#2{\underline{\box1}}\mkern#5mu$}}
\fontdimen8#13=\dimen@}
\makeatother

\newcommand{\Rbarplus}{\skewbar{\re}{-1}{-2}{0}_+}
\makeatother

\allowdisplaybreaks

\KOMAoptions{DIV=12}

\usepackage[pagebackref=false]{hyperref}
\hypersetup{colorlinks=true,linkcolor=blue,citecolor=blue,filecolor=blue,urlcolor=blue,pdfauthor={Aleksandar Arandjelovic}}
\hypersetup{final}

\setkomafont{disposition}{\bfseries}
\usepackage[onehalfspacing]{setspace}

\begin{document}

\title{Importance sampling for option pricing \\ with feedforward neural networks}

\author[1,2]{Aleksandar~Arandjelovi\'{c}%
\thanks{Corresponding author: \href{mailto:aleksandar.arandjelovic@fam.tuwien.ac.at}{aleksandar.arandjelovic@fam.tuwien.ac.at}}}
\author[1]{Thorsten~Rheinl{\"a}nder}
\author[2]{Pavel~V.~Shevchenko}

\affil[1]{Institute of Statistics and Mathematical Methods in Economics, TU Wien, Vienna, Austria}
\affil[2]{Department of Actuarial Studies and Business Analytics, Macquarie University, Sydney, Australia}

\date{\today}

\maketitle
\thispagestyle{empty}

\begin{abstract}%
\noindent We study the problem of reducing the variance of Monte Carlo estimators through performing suitable changes of the sampling measure which are induced by feedforward neural networks.
To this end, building on the concept of vector stochastic integration, we characterize the Cameron--Martin spaces of a large class of Gaussian measures which are induced by vector-valued continuous local martingales with deterministic covariation.
We prove that feedforward neural networks enjoy, up to an isometry, the universal approximation property in these topological spaces.
We then prove that sampling measures which are generated by feedforward neural networks can approximate the optimal sampling measure arbitrarily well.
We conclude with a comprehensive numerical study pricing path-dependent European options for asset price models that incorporate factors such as changing business activity, knock-out barriers, dynamic correlations, and high-dimensional baskets.

\smallskip\noindent Keywords: Cameron--Martin space, Dol\'{e}ans exponential, feedforward neural networks, importance sampling, universal approximation

\smallskip\noindent Subjclass: 60G15, 65B99, 65C05, 68T07, 91G20, 91G60
\end{abstract}

\clearpage

\section{Introduction}\label{sec:intro}

Monte Carlo methods are amongst the most essential tools for the numerical evaluation of financial derivatives.
Classical asset pricing theory often times calls for the computation of expectations of the form
\begin{equation*}
\mathbb{E}_{\PP}[F(X)] = \int_{\Omega} F(X(\omega))\, \PP(\mathrm{d}\omega),
\end{equation*}
where $F$ is a payoff functional, $X$ is a (multivariate) asset price process that is given as the solution to a stochastic differential equation (SDE) of the form $\mathrm{d}X_{t} = a(X)\, \mathrm{d}C_{t} + b(X)\, \mathrm{d}M_{t}$ for $t \in T = [\4 0, u]$ with a finite time horizon $u > 0$ for some potentially path-dependent coefficients $a$, $b$, a process $C$ of locally finite variation, and a local martingale $M$, and $\PP$ is a probability measure on a measurable space $(\Omega, \Fcal)$.
By averaging the payoffs over randomly sampled trajectories of $X$, one can estimate the price in many cases where no analytic solution for $\mathbb{E}_{\PP}[F(X)]$ is available.

The variance of the Monte Carlo estimator is inversely proportional to the number of trajectories simulated and proportional to the variance of the option payoff.
The square root of this variance is referred to as the standard error and in principle, it can be made as small as needed by simulating a sufficiently high amount of trajectories.
However, given limitations on computational time, the error can still be too large to be acceptable, especially for further calculations of option price derivatives (the so-called Greeks) required for hedging and risk management.

The usage of variance reduction methods can drastically reduce this error.
There are many different methods to reduce the variance of Monte Carlo estimators, one of which is importance sampling.
This method is based on changing the sampling measure from which the trajectories are generated from $\PP$ to some equivalent measure $\PP_{h}$, thereby overweighting important scenarios to increase the numerical efficiency of the estimates.
Due to Girsanov's theorem, this corresponds to adding a drift $h \in H$ to the process $M$, where $H$ denotes a prescribed space of functions or processes from which the drift adjustment is chosen.
One then writes
\begin{equation*}
\mathbb{E}_{\PP}[F(X)] = \mathbb{E}_{\PP_{h}}[F(X) Z_{h}^{-1}],
\end{equation*}
where $Z_{h}$ denotes a Radon--Nikod\'{y}m density of $\PP_{h}$ w.r.t. $\PP$, which also depends on the time-horizon $u$.
Instead of simulating realizations of $F(X)$ w.r.t. $\PP$, one then simulates realizations of $F(X) Z_{h}^{-1}$ w.r.t. $\PP_{h}$, and chooses $h \in H$ such that the variance of $F(X) Z_{h}^{-1}$ is minimized.

While this method usually requires a lot of specific knowledge about the model at hand, it has the potential to drastically reduce the variance of the corresponding Monte Carlo estimator.
In other words, importance sampling makes for a powerful method that involves the complex optimization problem of choosing an appropriate sampling measure which minimizes the variance of the Monte Carlo estimators.

Neural networks provide an algorithmically generated class of functions which on the one hand enjoy the universal approximation property in many different topological spaces, meaning that they are dense in these spaces, and on the other hand can be trained in a numerically efficient way.
Having recently entered the realm of mathematical finance, neural networks are successfully used e.g. for model calibration, hedging and pricing.
This paper develops a method that uses feedforward neural networks to perform importance sampling for complex stochastic models, which applies in particular to the evaluation of path-dependent derivatives.
By optimizing over drifts from a dense subspace $H(D)$ of $H$ that is generated by a set $D$ of feedforward neural networks, we obtain a tractable problem which is both theoretically justified and numerically efficient.

\subsection{Outline of the paper and main results}\label{subsec:outline}

In Section~\ref{sec:CMspace}, we characterize tractable spaces $H$ from which the drift adjustments may be chosen, and study their analytic properties.
Whenever $M$ is a vector-valued continuous local martingale with deterministic covariation, it induces a Gaussian measure, to which one can assign a Hilbert space $H$, the Cameron--Martin space.
Due to the multivariate nature of our study, we recall in Lemma~\ref{lem:L2space} some concepts which originate from the theory of stochastic integration with respect to vector-valued semimartingales.
A detailed characterization of the corresponding Cameron--Martin space $H$ that is induced by $M$ is provided in Proposition~\ref{prop:CMspace}, where the general formulation allows us to specifically incorporate complex and time-inhomogeneous covariance patterns for $M$ into our models.
Theorem~\ref{thm:CMapprox} then yields the essential approximation result that characterizes dense subspaces $H(D)$ of $H$ which are generated by prescribed sets of functions $D$ in an abstract and general setting, and in particular applies to sets $D$ of feedforward neural networks as a special case.

In Section~\ref{sec:ffn}, we focus our attention on feedforward neural networks, where we distinguish between neural networks of deep, narrow and shallow kind.
Propositions~\ref{prop:NNdense} and~\ref{prop:NNdense2} yield two approximation results which provide a theoretical justification for considering sets $D$ that consist of feedforward neural networks.
Example~\ref{ex:simple} then shows in a classical setting that the set $H(D)$ which is generated by a set $D$ of feedforward neural networks has an explicit and tractable characterization.
As a direct consequence of Theorem~\ref{thm:CMapprox}, Subsection~\ref{subsec:UATHoelder} then discusses a result which in particular implies that every smooth function can, up to an isometry, be approximated by feedforward neural networks arbitrarily well with respect to H{\"o}lder type topologies, which are stronger than the topology of uniform convergence.

Section~\ref{sec:IS} contains a detailed study of the importance sampling problem, where we aim to minimize the variance of $F(X) Z_{h}^{-1}$ w.r.t. $\PP_{h}$ by approximating the optimal drift $h$ with a feedforward neural network.
Theorem~\ref{thm:densityapprox} proves that the functional $V \from H \to \re_{+}$ which needs to be minimized is, under suitable generic assumptions, continuous and admitting a minimizer $h^{*} \in H$, which can be approximated, up to an isometry, arbitrarily well by feedforward neural networks.
Here, we not only prove convergence to the optimal drift $h^{*}$, but we moreover show that the corresponding Radon--Nikod\'{y}m densities converge to $Z_{h^{*}}$.
To this end, we prove that feedforward neural networks induce equivalent probability measures, whose densities with respect to the original measure converge in $L^{p}$-spaces, see Lemma~\ref{lem:density}.
Moreover, Subsection~\ref{subsec:UATGR} contains a discussion of a classical importance sampling approach that utilizes results from the theory of large deviations, where we show that feedforward neural networks can be employed to solve the corresponding variational problem which appears in this approach.
Let us note that, while the results from Sections~\ref{sec:CMspace} and~\ref{sec:ffn} are applied to importance sampling in Section~\ref{sec:IS}, they are also of independent interest.

Section~\ref{sec:simulation} contains a comprehensive numerical study pricing path-dependent European options for asset price models that incorporate factors such as changing business activity, knock-out barriers, dynamic correlations, and high-dimensional baskets.
To conclude, we summarize our findings and give an outlook on future work in Section~\ref{sec:conclusion}.
The appendix contains a brief glimpse at the theory of Gaussian measures and collects the proofs of all results.

\subsection{Related literature}\label{subsec:lit}

The line of research which eventually lead up to the present work originates from \cite{MR1849001}.
The authors study the problem of pricing path dependent options by using techniques from the theory of large deviations to perform a change of sampling measure that reduces the standard error of the Monte Carlo estimator.
Moreover, the authors use stratified sampling in order to further improve their simulations, and while we will not be using this technique, the interested reader might try to add stratified sampling on top of the method which we outline below.

The main motivation for this work was provided by Paolo Guasoni \& Scott Robertson in \cite{MR2362149}.
As in \cite{MR1849001}, the authors employ methods from the theory of large deviations to obtain a variational problem whose solution yields an asymptotically optimal drift adjustment.
The main difference to the present work is that we do not pass to a small noise limit. 
However, as it turns out, our method also complements the method presented in \cite{MR2362149}, see Subsection \ref{subsec:UATGR} for further details.

An extension of the methods used in \cite{MR2362149} to the study of importance sampling for stochastic volatility models has been provided in \cite{MR2565852}.
Note that our method does apply to these types of models as well, see Example~\ref{ex:stochvol} in Section~\ref{sec:IS} and Section~\ref{sec:simulation}, where we provide simulation results for several stochastic volatility models.

Another interesting contribution is \cite{dRST2018}, where the authors study importance sampling for McKean--Vlasov SDEs.
Similarly as in \cite{MR2362149,MR2565852} methods from the theory of large deviations yield an asymptotically optimal drift adjustment, and the authors discuss two different methods for the simulation of the solution to the McKean--Vlasov SDE under a change of measure.

The idea to use methods from the theory of stochastic approximation for the purpose of importance sampling has been studied extensively in \cite{MR2680557}.
This paper heavily influenced Section~\ref{sec:IS}, especially the proof of Theorem~\ref{thm:densityapprox} relies partially on a straightforward extension of the proof of~\cite[Proposition 4]{MR2680557}.
Let us also note that, while the setting of \cite[Section~3]{MR2680557} could be extended to our setting below, it might be of particular interest to understand how the algorithm proposed in \cite[Theorem~4]{MR2680557} could be adapted to the setting of Section \ref{sec:IS} below in order to yield convergence of the stochastic gradient descent algorithm when training feedforward neural networks.

Finally, one very original contribution that studies measure changes which are induced by neural networks for the purpose of Monte Carlo simulations is \cite{MMRGN19}.
The authors also study importance sampling and apply their results to light-transport simulations.
The main difference to \cite{MMRGN19} is that our method applies to the pricing of financial derivatives in a mathematically more natural way by using methods from the theory of stochastic calculus.
Here, we focus on neural networks that are of feedforward type.
For more details on the studied neural networks architectures and related literature, see Section \ref{sec:ffn}.

\begin{notation}
Unless stated otherwise, we endow $\re^{d}$ for each $d \in \na$ with the corresponding Euclidean norm $|\4\cdot\4|$.
$I_{d}$ denotes the identity matrix in $\re^{d \times d}$, and we write $\Rbarplus = \replus \cup \{ +\infty \}$.
Given two vectors $x,y$ of the same dimension, we denote by $x \odot y$ their Hadamard product.
If $\Sigma$ is a matrix, we denote by $\Sigma^{\top}$ its transpose.
For $x \in \re_{+}^{d}$ and $p > 0$, we understand $x^{p}$ to hold componentwise, and write $\sqrt{x}$ if $p = 1/2$.
Let us also convene that $\inf\{\emptyset\} = \infty$.
For each linear operator $A$ between normed spaces, we denote by $\| A \|_{\mathrm{op}}$ its operator norm.
If $E_{1}, E_{2}$ denote two metric spaces and $D$ is a subset of $E_{1}$, we say that $D$ is dense in $E_{2}$ up to an isometry, if there exists an isometry $J \from E_{1} \to E_{2}$, such that $J(D)$ is dense in $E_{2}$.
If $H$ denotes a Hilbert space, the notation $H^{*} \cong H$ is to indicate that we identify $H^{*}$ with $H$ via the isometric isomorphism given by Fr{\'e}chet--Riesz's representation theorem.

Given a topological space $(S, \mathcal{T})$, we denote by $S^{*}$ and $S'$ the topological and algebraic dual spaces, respectively, and write $(f, x) = f(x)$ for $(f,x) \in S^{*} \times S$ as well as $\B_{S}$ for the Borel $\sigma$-algebra on $S$.
Given $F \in \mathcal{T}$, we denote by $F^{\circ}$ and $\overbar{F}$ the interior and closure of $F$, respectively.
Whenever $\nu$ denotes a Borel measure on $S$, we say that $f \from S \to \re^{d}$ is locally $\nu$-essentially bounded, if $(\nu-)\esssup_{x \in K}|f(x)| < \infty$ for each compact $K \subset S$.
Given an interval $T=[\4 0, u]$, the space $C_{0}(T; \re^{d})$ consists of all $\re^{d}$-valued, continuous functions on $T$ that vanish at the origin.

Whenever $(S, \mathcal{S})$ is a measurable space, where $S$ denotes a set and $\mathcal{S}$ denotes a $\sigma$-algebra on $S$, we denote by $\mathcal{L}^{0}(\mathcal{S};\re^{d})$ the space of $\re^{d}$-valued, $\mathcal{S}$-measurable functions on $S$.
If $\mu$ is a measure on $\mathcal{S}$ and $f \in \mathcal{L}^{0}(\mathcal{S}) = \mathcal{L}^{0}(\mathcal{S};\re)$, we denote by $f \cdot \mu$ the Lebesgue integral of $f$ w.r.t. $\mu$, provided that it exists.
For $p > 0$, we further denote by $L^{p}(\mu)$ the space of equivalence classes of $p$-integrable functions from $\mathcal{L}^{0}(\mathcal{S})$.
The law of a random variable $Z$ is denoted by $\mathcal{L}(Z)$.
If $M$ is an $\re^{d}$-valued semimartingale, and $H \in L(M)$, we denote by $H^{\top} \sbullet[.75]M$ the stochastic integral of $H$ w.r.t. $M$.
We denote by $\mathcal{N}(m, \Sigma)$ the normal distribution with expected value $m \in \re^{d}$ and covariance matrix $\Sigma \in \re^{d \times d}$.
Finally, we denote for each $p \ge 1$ by $\mathcal{H}^{p}$ the Banach space of continuous $L^{p}$-integrable martingales, where the dependency on the underlying filtered probability space is implicit.
\end{notation}

\section{Universal approximation in Cameron--Martin space}\label{sec:CMspace}

In this section, we study a tractable space $H$ whose elements will be used to adjust the drift of $M$ for the purpose of importance sampling in Sections~\ref{sec:IS} and~\ref{sec:simulation} below.
Moreover, we identify dense linear subspaces of $H$ and obtain an explicit characterization of the Cameron--Martin spaces of a large class of Gaussian measures, which is of independent interest.
For details about Gaussian measures, we refer to Appendix~\ref{ap:GM}.

Let $(\Omega, \Fcal, \FF, \PP)$ with $\FF = (\Fcal_{t})_{t \in T}$ denote a filtered probability space, such that $\FF$ contains all $\PP$-null sets of $\Fcal$.
As index set for the time parameter, we consider $T =  [\4 0, u]$ with a finite time horizon $u > 0$.
Without loss of generality, we may assume that $\Fcal_{u} = \Fcal$.
We denote by $\lambda$ the restriction of the Lebesgue--Borel measure to $T$, and fix two dimensions $d, n \in \mathbb{N}$.
Let $M = (M_{t})_{t \in T}$ be an $\re^{d}$-valued continuous local martingale with $M_{0} = 0$.
Unless stated otherwise, we assume all stochastic processes to be $\FF$-adapted.

Let us start with a classical example that highlights the main concepts which are of importance in this section, before extending the study to a more general setting.

\begin{example}[Classical Wiener space]\label{ex:CWS}
Let $(E, H, \gamma)$ denote the classical Wiener space, where $E = C_{0}(T; \re^{d})$, $H = \bigl\{ h(t) = (\indicatorset{[\4 0, t]}f_{h})\cdot \lambda,\ t \in T\, |\, f_{h} \in L^{2}(\lambda; \re^{d})\bigr\}$ is the space of $\re^{d}$-valued, absolutely continuous functions on $T$ that admit a square-integrable density w.r.t. $\lambda$, and $\gamma$ is the classical Wiener measure on $E$, which is the Borel probability measure on $E$ that is induced by $\re^{d}$-valued standard Brownian motion $B = (B_{t})_{t \in T}$.

Endowed with the inner product $\langle g,h \rangle_{H} = \langle f_{g}, f_{h} \rangle_{L^{2}(\lambda; \re^{d})}$, one can show that $(H, \langle\cdot, \cdot \rangle_{H})$ is a real separable Hilbert space, which is moreover continuously embedded into $E$ as a dense linear subspace.
The operator $J \from L^{2}(\lambda; \re^{d}) \to H, f_{h} \mapsto h(\cdot) = (\indicatorset{[\4 0, \cdot]}f_{h})\cdot \lambda$ is a linear isometry by construction and thus continuous.
Whenever $D$ is a dense linear subspace of $L^{2}(\lambda; \re^{d})$, it follows that $J(D)$ is a dense linear subspace of $H$ and thus densely embedded into $E$.
In other words, $\overbar{J(D)} = H$ and $\overbar{J(D)} = E$, where the closure of $J(D)$ is taken in $H$ and $E$, respectively.
\end{example}

Section~\ref{sec:CMspace} is dedicated to a refined study of the identity $\overbar{J(D)} = H$ in a generalized setting.
To this end, let us state an assumption which allows us to study the process $M$ as a Gaussian process, and simplifies the proofs of Section~\ref{sec:IS}.
Moreover, it leads to a natural candidate for the space $H$ of drift adjustments, which consists of deterministic functions (see Definition~\ref{def:CMspace} below).

\begin{assumption}\label{assumption1}
The covariation process $[M]$ is, up to indistinguishability, deterministic, and $\trace([M])_{u} > 0$ outside a $\PP$-null set.
\end{assumption}

In what follows, we disregard the evanescent- and $\PP$-null sets on which the two conditions from Assumption~\ref{assumption1} are violated, and consider equalities between stochastic processes and (in)equalities between random variables to hold up to indistinguishability and $\PP$-almost surely, respectively.

\begin{definition}
The quadratic variation process $C \coloneqq \trace([M])$, being increasing and of finite variation, induces a finite Lebesgue--Stieltjes measure on $(T, \B_{T})$, which we denote $\mu$.
\end{definition}

\begin{remark}[Elementary properties of $M$]\label{rem:Gprop}
Due to L\'{e}vy's characterization (cf. \cite[Theorem~7.1]{Schmock2021}), the increments $M_{t} - M_{s}$ are independent of $\Fcal_{s}$ with $\mathcal{L}(M_{t} - M_{s}) = \N(0, [M]_{t} - [M]_{s})$ for all $s < t$ in $T$.
Therefore, $M$ is a centered Gaussian process, and \cite[Theorem~11.5]{MR4226142} shows that $M$ is an $\FF$-Markov process.
$M$ is actually a martingale, since $\mathbb{E}[ M_{t}-M_{s}\, |\, \Fcal_{s}] = \mathbb{E}[ M_{t}-M_{s}] = 0$ for $s < t$ in $T$.
Note that $\Cov(M_{t}, M_{s}) = [M]_{s \wedge t}$ for $s,t \in T$ since, assuming without loss of generality that $s < t$,
\begin{equation*}
\Cov(M_{t}, M_{s}) = \mathbb{E}[M_{t} M_{s}^{\top}] = \mathbb{E}[(M_{t}-M_{s}) M_{s}^{\top}] + \Cov(M_{s}, M_{s}) = [M]_{s}.
\end{equation*}
\end{remark}

\begin{remark}\label{rem:covariation}
The non-degeneracy condition from Assumption~\ref{assumption1} implies that $\mu(T) > 0$.

We write $\mu = \mu_{\mathrm{a}} + \mu_{\mathrm{s}}$ for the Lebesgue decomposition of $\mu$ w.r.t. $\lambda$ into an absolutely continuous measure $\mu_{\mathrm{a}} = f_{\lambda} \cdot \lambda$ and a singular measure $\mu_{\mathrm{s}}$, where $f_{\lambda}$ denotes a Radon--Nikod\'{y}m density of $\mu_{\mathrm{a}}$ w.r.t. $\lambda$, and both $\mu_{\mathrm{a}}$ and $\mu_{\mathrm{s}}$ are finite measures.
Note that $\mu$ has no atoms, since $[M]$ and therefore also $C$ are continuous.

We now proceed in line with \cite[Section~12.5]{MR3443368}.
For $i,j \in \{1,2,\hdots,d\}$, the covariation process $[M^{i}, M^{j}]$, being of finite variation, induces a finite signed (and due to the continuity of $[M^{i}, M^{j}]$ atomless) measure $\mu_{i,j}$ on $(T, \B_{T})$.
It follows from the Kunita--Watanabe inequality for Lebesgue--Stieltjes integrals (cf. \cite[Lemma~5.89]{Schmock2021}), that the total variation measure $|\mu_{i,j}|$ is absolutely continuous w.r.t. $\mu$, hence an application of the Radon--Nikod\'{y}m theorem for signed measures (cf. \cite[Section~1.7.14]{MR3443368}) yields the existence of a real-valued density $\mathrm{d}\mu_{i,j} / \mathrm{d}\mu \eqqcolon \pi_{i,j}$, that is in $L^{1}(\mu)$ since $\mu_{i,j}$ is finite.

We collect $(\pi_{i,j})_{i,j=1, \hdots, d}$ into a measurable function $\pi$ that assumes, due to the symmetry of $[M]$, values in the space of symmetric matrices in $\re^{d \times d}$, and write
\begin{equation}\label{eq:covarid}
[M]_{t} = \bigl(\pi \sbullet[.75] C\bigr)_{t} = (\indicatorset{[\4 0,t]}\pi) \cdot \mu,
\quad t \in T,
\end{equation}
where the notation is to be understood componentwise.
Let $(\eta_{k})_{k \in \na}$ be dense in $\re^{d}$.
For each $k \in \na$, set $A_{k} = \{s \in T\, |\, \eta_{k}^{\top}\pi(s)\eta_{k} \ge 0\}$ as well as $A = \bigcap_{k \in \na} A_{k} = \{s \in T\, |\, \eta^{\top} \pi(s) \eta \ge 0,\, \forall\, \eta \in \re^{d}\}$.
Note that
\begin{equation*}
\int_{[\4 0, t]} \eta_{k}^{\top} \pi(s) \eta_{k}\, \mu(\mathrm{d}s) = \bigl( (\eta_{k}^{\top} \pi \eta_{k}) \sbullet[.75] C\bigr)_{t} = [\eta_{k}^{\top} M]_{t} \ge 0,
\quad k \in \na,\, t \in T,
\end{equation*}
which implies that each $A_{k}^{\mathrm{c}}$ is a $\mu$-null set, and therefore $A^{\mathrm{c}}$, being the countable union of all sets $A_{k}^{\mathrm{c}}$, is a $\mu$-null set, too.
We conclude that $\pi$ is positive semi-definite $\mu$-a.e.
Note that we could, in the spirit of \cite{MR1975582,MR542115,MR580121,MR568256} and without loss of generality, replace $\pi$ by $\tilde{\pi} = \pi \indicatorset{A}$, and thus assume that it is positive semi-definite for each $t \in T$.
For the purpose of this paper, this step is not necessary though.
\end{remark}

\begin{remark}[Non-uniqueness of $(\pi, C)$]
The decomposition of $[M]$ into a matrix-valued function $\pi$ and an increasing process $C$ is not unique.
For example, take $\tilde{C} = \sum_{i=1}^{d} \eta_{i} [M^{i}]$, where $\eta \in \re^{d}$ is chosen such that $\eta_{i} > 0$ for all $i \in \{1,2,\hdots, d\}$.
More generally, take $\tilde{C} = \sum_{i=1}^{d} f_{i} \sbullet[.75] [M^{i}]$, with $f_{i} \from T \to \replus \setminus \{0\}$ in $L^{1}(\mu_{i,i})$ for each $i \in \{1,2,\hdots, d\}$.
In both cases, the corresponding function $\tilde{\pi}$ is then constructed as in Remark~\ref{rem:covariation}, and generally differs from $\pi$.
Lemma~\ref{lem:L2space5} below will show that the non-uniqueness of $(\pi, C)$ is not a problem though.
\end{remark}

\begin{example}\label{ex:BM}
If $\pi \equiv I_{d}$ and $\mu = \lambda$, then $M$ is, by L\'{e}vy's characterization, an $\re^{d}$-valued standard Brownian motion.
\end{example}

\begin{example}[Multivariate Heston model]\label{ex:MVHeston}
An example that holds relevance for practitioners is the multivariate Heston model, which we will briefly describe.

Let $d=2n$ for some $n \in \mathbb{N}$.
We consider a dynamic diffusion matrix given by $T \ni t \mapsto \Sigma(t) \in \mathbb{R}^{d \times d}$, a vector of appreciation rates $r \in \mathbb{R}^{n}$, a vector of mean-reversion levels $m \in \mathbb{R}^{n}$, and a diagonal matrix $\Theta \in \mathbb{R}^{n \times n}$ representing mean-reversion speeds. 
To avoid degeneracy, we assume that $(\Sigma_{k, \cdot}(t))^{\top}$ is not the zero vector for each $k \in {1,2,\dots,n}$ and $t \in T$.
Let $M_{t} = \Sigma(t) B_{t}$, where $B$ denotes a standard Brownian motion with values in $\mathbb{R}^{d}$. 
Note that $[M]_{t} = \mathrm{Cov}(M_{t}, M_{t}) = \int_{0}^{t}\Sigma(s) \Sigma^{\top}(s)\, \mathrm{d}s$ for each $t \in T$.
Hence, in light of Remark \ref{rem:covariation} above, we may choose $\mu = \lambda$ and $\pi(t) = \Sigma(t) \Sigma^{\top}(t)$.

Fix $s,x \in \re_{+}^{n}$.
For notational simplicity, we write $M^{(1)} = (M^{1}, M^{2}, \hdots, M^{n})^{\top}$ as well as $M^{(2)} = (M^{n+1}, M^{n+2}, \hdots, M^{2n})^{\top}$ such that $M = (M^{(1)}, M^{(2)})^{\top}$.
Write $X = (S,V)^{\top}$ and let the asset price follow the SDE $\mathrm{d}S_{t} = (r \odot S_{t})\, \mathrm{d}t + (S_{t} \odot \sqrt{V_{t}}) \odot \mathrm{d}M_{t}^{(1)}$, subject to $S_{0} \equiv s$.
The $n$-dimensional instantaneous variance process $V$ follows the Cox--Ingersoll--Ross (CIR) type SDE $\mathrm{d}V_{t} = \Theta(m-V_{t})\, \mathrm{d}t + \sqrt{V_{t}} \odot \mathrm{d}M_{t}^{(2)}$, subject to $V_{0} \equiv v$.
Here, we see that asset price models whose dynamics are driven by multivariate Brownian motions with dynamic variance-covariance matrices fall within the scope of our setting.
More generally, one could think of replacing $B_{t}$ by a time-changed Brownian motion $B_{f(t)}$ for a given deterministic time-change $f$.

In Section \ref{sec:simulation}, we will study special cases of this model, where we will impose either a time-change to model changing levels of business activity, or a dynamic correlation structure.
\end{example}

\begin{example}[An example where $\mu$ is singular w.r.t. $\lambda$]\label{ex:pathological}
Set $u = 1$, and let $B$ be a standard $\re^{d}$-valued $(\mathbb{G}, \PP)$-Brownian motion, where $\mathbb{G} = (\mathcal{G}_{t})_{t \in T}$ denotes a filtration of $\Fcal$.
Let $f \from T \to T$ be either Cantor's ternary function or Minkowski's question-mark function, and let $\Sigma \in \re^{d \times d}$ be a diffusion matrix.
Recall that Cantor's ternary function is continuous, monotonically increasing, has derivative zero on a set of Lebesgue measure zero, but is not absolutely continuous.
Likewise, Minkowski's question-mark function has the same properties, while being even strictly increasing.
Set $M_{t} \coloneqq \Sigma B_{f(t)}$ for $t \in T$, and note that $M$ is an $\re^{d}$-valued continuous $(\FF, \PP)$-martingale with $M_{0} = 0$ and $[M]_{t} = f(t)\Sigma \Sigma^{\top}$ for $t \in T$, where the filtration $\FF = (\Fcal_{t})_{t \in T}$ is given by $\Fcal_{t} = \mathcal{G}_{f(t)}$ for $t \in T$.
The corresponding Lebesgue--Stieltjes measure $\mu$ is singular w.r.t. $\lambda$, and while $\trace([M])$ is increasing when $f$ is Cantor's ternary function, $\trace([M])$ is even strictly increasing when $f$ is Minkowski's question-mark function.
See \cite{MR7929} for further examples of functions $f$ that can be used for constructions of this kind.
\end{example}

Based on the pair $(\pi, \mu)$, we define a weighted $L^{2}$-space, which we denote $\Lambda^{2}$, which is a generalization of the space $L^{2}(\lambda; \re^{d})$ in the context of Example~\ref{ex:CWS}, and recall some elementary properties.
In Section~\ref{sec:IS}, where we study importance sampling, functions $f$ from $\Lambda^{2}$ will be used to construct equivalent measures via the Dol\'{e}ans exponential $\mathcal{E}(f^{\top}\sbullet[.75]M)$.
As we will argue in Section \ref{sec:ffn}, feedforward neural networks are dense in $\Lambda^{2}$ under suitable assumptions which, due to Theorem~\ref{thm:densityapprox}, provides a theoretical justification for using feedforward neural networks in order to calibrate an optimal sampling measure that minimizes the variance of the Monte Carlo estimators in Sections~\ref{sec:IS} and~\ref{sec:simulation}.
The definition of the space $\Lambda^{2}$ is inspired by the concept of vector stochastic integration, see \cite{MR1975582,MR542115,MR568256}, and in particular \cite[Chapitre~{\MakeUppercase{\romannumeral 4}}]{MR542115}, for further details and generalizations.

\begin{lemma}[D'apr{\`e}s Jacod]\label{lem:L2space}
Let\/ $\Lambda^{2}$ denote the set of all\/ $f \in \mathcal{L}^{0}(\B_{T};\re^{d})$ with
\begin{equation*}
\| f \|_{\Lambda^{2}} \coloneqq \Bigl(\int_{T} f^{\top}(s)\pi(s)f(s)\, \mu(\mathrm{d}s)\Bigr)^{1/2} < \infty,
\end{equation*}
where we identify\/ $f,g \in \Lambda^{2}$ if\/ $(f-g)^{\top}\pi(f-g) =0$\/ $\mu$-a.e., and write\/ $f \sim g$ in this case.
We further set\/ $\langle f,g \rangle_{\Lambda^{2}} \coloneqq \int_{T} f^{\top}(s)\pi(s)g(s)\, \mu(\mathrm{d}s)$ for\/ $f,g \in \Lambda^{2}$.
Then:

\begin{enumerate}[ref={\thelemma(\alph*)}, label=(\alph*)]
\setlength\itemsep{0.1em}
\item\label{lem:L2space1}
$(\Lambda^{2}, \langle \cdot, \cdot \rangle_{\Lambda^{2}})$ is a real separable Hilbert space;

\item\label{lem:L2space2}
To each\/ $F \in (\Lambda^{2})^{*}$ there corresponds a unique function\/ $g \in \Lambda^{2}$, such that
\begin{equation*}
F(f) = \int_{T} g^{\top}(s)\pi(s)f(s)\, \mu(\mathrm{d}s),
\quad f \in \Lambda^{2},
\end{equation*}
and $\| F \|_{\mathrm{op}} = \|g\|_{\Lambda^{2}}$.
Therefore,\/ $(\Lambda^{2})^{*}$ is isometrically isomorphic to\/ $\Lambda^{2}$;

\item\label{lem:L2space3}
We denote by\/ $\Lambda^{2,0}$ the set of all\/ $f \in \mathcal{L}^{0}(\B_{T};\re^{d})$ that satisfy\/ $f_{i} \in L^{2}(\mu_{i,i})$ for each\/ $i \in \{1,2,\hdots,d\}$, where we identify functions in the same manner as above.
Then:

\begin{enumerate}[label=(\arabic*), ref=\thelemma\labelenumi]
\setlength\itemsep{0.1em}
\item\label{lem:L2space3a}
$(\Lambda^{2,0},\langle \cdot, \cdot \rangle_{\Lambda^{2}})$ is a separable inner product space with\/ $\Lambda^{2,0} \subset \Lambda^{2}$;
\item\label{lem:L2space3b}
$\Lambda^{2,0}$ is dense in\/ $\Lambda^{2}$, hence\/ $\Lambda^{2}$ is the completion of\/ $\Lambda^{2,0}$ w.r.t.\/ $\|\4\cdot\4 \|_{\Lambda^{2}}$;
\end{enumerate}

\item\label{lem:L2space4}
$(C(T; \re^{d}), \|\4\cdot\4\|_{\infty})$ is continuously embedded into\/ $(\Lambda^{2,0}, \|\4\cdot\4\|_{\Lambda^{2}})$ as a dense linear subspace, where\/ $\|f\|_{\infty} \coloneqq \sup_{t \in T}|f(t)|$ for\/ $f \in C(T; \re^{d})$;

\item\label{lem:L2space5}
$\Lambda^{2}$ and\/ $\Lambda^{2,0}$ do not depend on the specific choice of\/ $(\pi, \mu)$ that satisfy \eqref{eq:covarid}.
\end{enumerate}
\end{lemma}

\begin{example}
According to \cite[Lemme~4.30]{MR542115} and the discussion thereafter, a sufficient condition for $\Lambda^{2,0} = \Lambda^{2}$ to hold is if there exists a constant $c > 0$ such that $\sum_{i=1}^{d}\pi_{i,i}f_{i}^{2} \le c f^{\top}\pi f$\/ $\mu$-a.e. holds for all $f \in \Lambda^{2}$.
Examples where this applies are when $\pi$ is a diagonal or uniformly strictly elliptic matrix, where the latter condition means that there exists a constant $c > 0$ such that $c |\eta|^{2} \le \eta^{\top}\pi \eta$ holds for all $\eta \in \re^{d}$.
\end{example}

\begin{example}\label{ex:equivalenceclasses}
Let us state one example, which is a deterministic version of \cite[Example~12.5.1]{MR3443368}, where $\Lambda^{2,0} \neq \Lambda^{2}$.
To this end, let $B$ denote a real-valued standard Brownian motion.
Set $M = (B, -B)^{\top}$ and note that $M$ is an $\re^{2}$-valued continuous martingale with covariation
\begin{equation*}
[M]_{t} =
\begin{pmatrix}
t & -t \\
-t & t \\
\end{pmatrix},
\quad \mathrm{hence} \quad
\pi \equiv
\begin{pmatrix}
1 & -1 \\
-1 & 1 \\
\end{pmatrix},
\end{equation*}
where we choose $\mu = \lambda$.

$\pi$ is positive semi-definite, since for each $\eta \in \re^{2}$, we have $\eta^{\top}\pi\eta = (\eta_{1}^{2}+\eta_{2}^{2}-2\eta_{1}\eta_{2}) = (\eta_{1}-\eta_{2})^{2}$, which is zero precisely when $\eta_{1} = \eta_{2}$, and positive otherwise.
Let $f \from T \to \re$ be measurable and such that $f \not\in L^{2}(\lambda)$.
Consider the function $g \from T \to \re^{2}$ given by $g = (f,f)^{\top}$.
By construction, we then have $g \not\in \Lambda^{2,0}$, but since $\| g \|_{\Lambda^{2}} = 0$, we have $g \in \Lambda^{2}$.
\end{example}

For $f \in \Lambda^{2}$, we have $f \in L^{2}(M)$ in the sense of vector stochastic integration (see also Lemma~\ref{lem:IKWapprox} in Appendix~\ref{app:technical}).
Since $[f^{\top}\sbullet[.75]M]_{u} = \|f\|_{\Lambda^{2}}^{2}$ is deterministic and finite, Novikov's criterion shows that $Z = \mathcal{E}(f^{\top} \sbullet[.75] M)$ is a strictly positive uniformly integrable martingale.
Girsanov's theorem shows that under the measure $\mathbb{Q}$ with $\mathrm{d}\mathbb{Q} / \mathrm{d}\PP = Z_{t}$ on $\Fcal_{t}$ for each $t \in T$, the finite variation part in the semimartingale decomposition of $M$ is given by $[f^{\top}\sbullet[.75]M, M]  = h$, where $h(t) = (\indicatorset{[\4 0, t]} \pi f) \cdot \mu$ for $t \in T$.
These considerations motivate the following definition.

\begin{definition}\label{def:CMspace}
We denote by $H$ the set of all $h \from T \to \re^{d}$ with the representation
\begin{equation}\label{CMfunc}
h(t) = J(f_{h})(t) \coloneqq \int_{[0, t]}\pi(s)f_{h}(s)\, \mu(\mathrm{d}s),
\quad t \in T,
\end{equation}
for some $f_{h} \in \Lambda^{2}$, where the integral in \eqref{CMfunc} is to be understood componentwise as a Lebesgue--Stieltjes integral.
\end{definition}

As Proposition~\ref{prop:CMspace} below will show, upon being endowed with an appropriate inner product, $H$ becomes the Cameron--Martin space of the Gaussian measure $\gamma_{M}$ which is induced by $M$ on $C_{0}(T; \re^{d})$.
To the best of our knowledge, there exists no explicit characterization of the Cameron--Martin space of $\gamma_{M}$ at the present level of generality in the literature so far, as one usually assumes $M$ to be a Brownian motion, which is a special case of our setting (see Example~\ref{ex:BM}).

\begin{example}
In the context of Example~\ref{ex:BM}, $H$ coincides with the set of absolutely continuous functions whose densities are square-integrable w.r.t. $\lambda$.
\end{example}

\begin{remark}[Extension to multivariate Volterra type Gaussian processes]\label{rem:fbm}
The Cameron--Martin space of fractional Brownian motion is not contained in our framework, except for the special case of a Brownian motion.
The matrix-valued function $\pi$ is not to be confused with the square-integrable but singular kernel which appears in integral representations of fractional Brownian motion and, more generally, Volterra type Gaussian processes.
However, our framework can be extended to multivariate versions of these processes with representations of the form $\tilde{M}_{t} = \int_{0}^{u}k(t,s)\, \mathrm{d}M_{s}$, where $k$ denotes an $\re^{d \times d}$-valued kernel function, for which, under suitable assumptions on $k$, the corresponding Cameron--Martin space consists of functions of the form
\begin{equation*}
\tilde{h}(t) = \int_{[0, u]}k(t,s)\pi(s)f_{h}(s)\, \mu(\mathrm{d}s),
\quad t \in T,
\end{equation*}
for some $f_{h} \in \Lambda^{2}$.
This formulation gives rise to the study of refined versions of multivariate Volterra type Gaussian processes as well as multivariate fractional stochastic volatility models, as one can now distinguish more explicitly between time-inhomogeneous volatility patterns which are induced by $\mu$ (or equivalently, by the quadratic variation $C$), the dependency structure of the components of $M$, which is modeled by the function $\pi$, and the path irregularities of $\tilde{M}$, which are induced by the matrix-valued kernel $k$.
\end{remark}

\begin{remark}
Equation~\eqref{CMfunc} suggest a generalization, where the functions $f_{h}$ assume values in a (possibly infinite-dimensional) Hilbert space $\tilde{H}$, and $\pi$ assumes $\mu$-a.e. values in the set of positive semi-definite operators on $\tilde{H}$.
In this case, the integral in~\eqref{CMfunc} is to be understood as a Lebesgue--Stieltjes--Bochner integral.
\end{remark}

\begin{example}[Forward curve spaces]
Let $w \from T \to [1, \infty)$ be a function in $L^{1}(\lambda; \re)$ that is non-decreasing, set $\Sigma = I_{d}$, and define the measure $\mu$ through $\mu(A) = \int_{A}w(s)\, \lambda(\mathrm{d}s)$ for $A \in \B_{T}$.
In line with Example~\ref{ex:pathological}, we can then construct a process $M$ by means of the increasing and continuous function $f \from T \to \replus$ given by $f(t) \coloneqq \mu\bigl([\4 0,t]\bigr)$ for $t \in T$.
In the context of Definition~\ref{def:CMspace}, the corresponding space $H$ then has similarities to the forward curve space $H_{w}$ (cf. \cite[Chapter~5]{MR1828523}) that is used in interest rate modeling.
\end{example}

Since elements from $H$ will be precisely those which we consider for the drift adjustment of $M$ in Sections~\ref{sec:IS} and~\ref{sec:simulation}, we need to collect some useful properties which will be needed later on (in particular for Theorem~\ref{thm:CMapprox}).
The following proposition collects these properties and further deepens the connections to the process $M$.
As it turns out, being endowed with a suitable inner product, $H$ is not only the isometric image of the space $\Lambda^{2}$ whose definition was inspired by the representation \eqref{eq:covarid} of $[M]$, but $H$ is also the Cameron--Martin space of the Gaussian measure $\gamma_{M}$ which is induced by $M$ on $C_{0}(T; \re^{d})$.

\begin{proposition}\label{prop:CMspace}
Consider the mapping\/ $\langle \cdot, \cdot \rangle_{H}$ on\/ $H \times H$ given by\/ $\langle g,h \rangle_{H} \coloneqq \langle f_{g}, f_{h} \rangle_{\Lambda^{2}}$.
Then:

\begin{enumerate}[ref={\theproposition(\alph*)}, label=(\alph*)]
\setlength\itemsep{0.1em}
\item\label{prop:CMspace1}
The integral in \eqref{CMfunc} is well defined for all\/ $f_{h} \in \Lambda^{2}$ and\/ $t \in T$;

\item\label{prop:CMspace2}
$(H, \langle\cdot,\cdot\rangle_{H})$ is a real separable Hilbert space;

\item\label{prop:CMspace3}
$J \from \Lambda^{2} \to H$ is a linear isometry, and\/ $(H^{0}, \langle\cdot,\cdot\rangle_{H})$ is an inner product space whose completion is\/ $(H, \langle\cdot,\cdot\rangle_{H})$, where we set\/ $H^{0} \coloneqq J(\Lambda^{2,0}) \subset H$;

\item\label{prop:CMspace5}
To each\/ $F \in H^{*}$ there corresponds a unique function\/ $g_{F} \in H$, such that
\begin{equation*}
F(h) = \langle g_{F}, h \rangle_{H} = \int_{T} f_{g_{F}}^{\top}(s)\pi(s)f_{h}(s)\, \mu(\mathrm{d}s),
\quad h \in H,
\end{equation*}
and $\| F \|_{\mathrm{op}} = \|g_{F}\|_{H}$.
Therefore,\/ $H^{*}$ is isometrically isomorphic to\/ $H$;

\item\label{prop:CMspace6}
$H$ is the Cameron--Martin space of the centered Gaussian measure\/ $\gamma_{M}$ that is induced by\/ $M$ on $E = C_{0}(T; \re^{d})$.
\end{enumerate}
\end{proposition}

\begin{remark}
Unless $\overbar{H} = E$, where the closure is taken in $E$, the measure $\gamma_{M}$ will be degenerate (see Remark~\ref{app:degenerate}).
The identity $\overbar{H} = E$ holds in some special cases, e.g. if $\mu = \lambda$ and $\pi \equiv I_{d}$, where the proof builds on the fact that continuous functions can be uniformly approximated by piecewise linear functions.
\end{remark}

For the purpose of the next result, we introduce the function $I \from E \to \Rbarplus$,
\begin{equation}\label{eq:ratefunc}
    I(g)=
    \begin{cases}
        \frac{1}{2}\int_{T}f_{g}^{\top}(s)\pi(s)f_{g}(s)\, \mu(\mathrm{d}s) & \mathrm{for}\, g \in H, \\
        \infty & \mathrm{otherwise}.
    \end{cases}
\end{equation}

\begin{example}
In the context of Example~\ref{ex:MVHeston}, the function $\pi$ takes the form $\pi(t) = \Sigma(t)\Sigma^{\top}(t)$.
On the other hand, in the context of Example~\ref{ex:pathological}, the measure $\mu$ could be the Lebesgue--Stieltjes measure that is induced by Cantor's ternary function and thus singular with respect to $\lambda$.
The setting typically discussed in the literature, where $M$ is a standard Brownian motion, does not encompass either of these examples.
\end{example}

In Subsection~\ref{subsec:UATGR} we will discuss an importance sampling method that uses methods from the theory of large deviations.
To this end, one needs to understand the asymptotic behavior of the scaled process $\sqrt{\varepsilon}M$ as $\varepsilon \searrow 0$.
If $M$ is a Brownian motion, then the corresponding result is referred to as Schilder's theorem (cf. \cite[Corollary~4.9.3]{MR1642391}, \cite[Theorem~8.3]{MR3024389} and \cite[Theorem~8.4.1]{MR2760872}).
As a consequence of Proposition~\ref{prop:CMspace6} and Proposition~\ref{app:LDP}, we obtain the following result, whose novelty is the explicit characterization of the function $I$ (also referred to as rate function) in \eqref{eq:ratefunc} at the presented level of generality.

\begin{proposition}
In the context of Proposition~\ref{prop:CMspace6} we have, for each $F \in \B_{E}$,
\begin{equation*}
    -\inf_{g \in F^{\circ}}I(g) \le \liminf_{\varepsilon \searrow 0}\varepsilon\log\PP\bigl(\sqrt{\varepsilon}M \in F\bigr) \le \limsup_{\varepsilon \searrow 0}\varepsilon\log\PP\bigl(\sqrt{\varepsilon}M \in F\bigr) \le -\inf_{g \in \overbar{F}}I(g)
\end{equation*}
where the function\/ $I \from E \to \Rbarplus$ is specified in \eqref{eq:ratefunc}.
\end{proposition}

For notational convenience, we introduce the following convention.

\begin{convention}\label{convention3}
Henceforth we denote by $D$ a dense subset of $\Lambda^{2}$.
\end{convention}

\begin{example}
We have already encountered two admissible candidates for the set $D$: $\Lambda^{2,0}$ and $C(T; \re^{d})$, see Lemma~\ref{lem:L2space3} and~\ref{lem:L2space4}.
\end{example}

The following theorem helps us to identify dense subsets and subspaces of $H$, and shows how these relate to the topological support (see Remark~\ref{app:degenerate} in Appendix~\ref{ap:GM}) of the measure $\gamma_{M}$ which we encountered in Proposition~\ref{prop:CMspace}.

\begin{theorem}\label{thm:CMapprox}
Let\/ $H(D)$ denote the set of all\/ $h \in H$ with\/ $f_{h} \in D$.
Then:

\begin{enumerate}[ref={\thetheorem(\alph*)}, label=(\alph*)]
\setlength\itemsep{0.1em}
\item\label{thm:CMapprox1}
$H(D)$ is a dense subset of\/ $H$ which is separable when endowed with the subspace topology;

\item\label{thm:CMapprox2}
If\/ $D$ is also a linear subspace of\/ $\Lambda^{2}$, then:

\begin{enumerate}[label=(\arabic*), ref=\thetheorem\labelenumi]
\setlength\itemsep{0.1em}
\item\label{thm:CMapprox2a}
$(H(D), \langle\cdot,\cdot\rangle_{H})$ is an inner product space, whose completion is\/ $H$;
\item\label{thm:CMapprox2b}
There exists a countable orthonormal basis of\/ $H$ which consist of elements from\/ $H(D)$;
\end{enumerate}

\item\label{thm:CMapprox3}
In the context of Proposition~\ref{prop:CMspace6}, the topological support of\/ $\gamma_{M}$ coincides with\/ $\overbar{H(D)}$, where the closure is taken in\/ $E$.
In other words,
\begin{equation*}
\gamma_{M}\bigl( C_{0}(T; \re^{d}) \setminus \overbar{H(D)} \bigr) = \mathbb{P}\bigl( M \in C_{0}(T; \re^{d}) \setminus \overbar{H(D)} \bigr) = 0,
\end{equation*}
hence outside a\/ $\PP$-null set, paths of\/ $M$ can be uniformly approximated by sequences from\/ $H(D)$.
\end{enumerate}
\end{theorem}

In Section~\ref{sec:ffn}, we will discuss classes of feedforward neural networks that are also dense subsets of $\Lambda^{2}$, thereby satisfying Convention~\ref{convention3}.
Together with Theorem~\ref{thm:CMapprox}, this will show that we can approximate any element from $H$, up to the isometry $J$, by feedforward neural networks, which will be essential for Sections~\ref{sec:IS} and~\ref{sec:simulation}, where we will approximate the drift adjustment of $M$ which minimizes the variance of the Monte Carlo estimator with feedforward neural networks.

\section{Approximation capabilities of feedforward neural networks}\label{sec:ffn}

In this section, we study feedforward neural networks as elements of the space $\Lambda^{2}$ and show how they generate, under suitable assumptions on the activation function, dense subspaces of $H$, thereby providing a first theoretical justification for approximating the optimal drift adjustment with feedforward neural networks when studying importance sampling in Sections~\ref{sec:IS} and~\ref{sec:simulation} below.

We know from Lemma~\ref{lem:L2space4} that $C(T; \re^{d})$ is continuously embedded into $\Lambda^{2,0}$ as a dense linear subspace.
Moreover, Lemma~\ref{lem:L2space3b} shows that $\Lambda^{2,0}$ is dense in $\Lambda^{2}$.
Consequently, every dense linear subspace $D$ of $C(T; \re^{d})$ is densely embedded into $\Lambda^{2}$, thereby satisfying Convention~\ref{convention3}, in which case $H(D)$ is dense in $H$ by Theorem~\ref{thm:CMapprox1}.
For this reason, we first focus our attention on finding dense linear subspaces of $C(T; \re^{d})$ in Proposition~\ref{prop:NNdense}, before relaxing the continuity assumption in Proposition~\ref{prop:NNdense2} below.

Neural networks are one particular class of functions which is of interest to us.
On the one hand it satisfies the required density in $C(T; \re^{d})$, and on the other hand it gives rise to efficient numerical optimization procedures which have lead to fascinating results in the domain of financial and actuarial mathematics in recent years.
When studying neural networks, the property of being dense in a topological space is referred to as the universal approximation property (UAP; cf. \cite[Definition~2]{Kra2021}).
Theorems which establish the density of neural networks in a topological space are referred to as universal approximation theorems (UAT).

There are many different neural network architectures for which the universal approximation property has been shown to hold in various topological spaces (cf. \cite{MR1015670,F1989,H1991,HSW1989,KL2020,LLPS1993,LFN2003,MM1992,PS1991,PS1993,P1999,ZWML1995}).
For ease of presentation, we will discuss the class of feedforward neural networks, also called multilayer perceptrons or multilayer feedforward neural networks.
Note that our discussion is limited to architectures that yield universal approximation theorems in $C(T; \re^{d})$ and $\Lambda^{2}$ (see also Subsection~\ref{subsec:UATHoelder} for a UAT w.r.t. Hölder norms).
There are many architectures for which the UAP has been established in other topological spaces, but which we will not discuss in this paper.

\begin{definition}[cf. {\cite[Definition~3.1]{KL2020}}]\label{def:fnn}
Given\/ $k \in \na$,\/ $l \in \na$ and\/  $\psi \from \re \to \re$, we denote by\/ $\mathcal{NN}_{k,l}^{d}(\psi)$ the set of feedforward neural networks with one neuron in the input layer,\/ $d$ neurons with identity activation function in the output layer,\/ $k$ hidden layers, and at most\/ $l$ hidden nodes with $\psi$ as activation function in each hidden layer.
\end{definition}

\begin{remark}
Elements from $\mathcal{NN}_{k,l}^{d}(\psi)$ can be represented as follows.
Consider affine functions $W_{1} \from \re \to \re^{l}$, $W_{k+1} \from \re^{l} \to \re^{d}$ and $W_{2}, \hdots, W_{k} \from \re^{l} \to \re^{l}$.
For $i=1, 2, \hdots, k$, denote $F_{i} = \psi \circ W_{i}$, where the activation function $\psi$ is applied componentwise.
Then, an element from $\mathcal{NN}_{k,l}^{d}(\psi)$ is given by the map $t \mapsto W_{k+1} \circ F_{k} \circ \cdots \circ F_{1}(t)$.
\end{remark}

If the number of nodes in the hidden layers can be arbitrarily large, we write $\mathcal{NN}_{k, \infty}^{d}(\psi)$.
Likewise, we write $\mathcal{NN}_{\infty, l}^{d}(\psi)$ if the number of hidden layers can be arbitrarily large.
Finally, the notation $\mathcal{NN}_{\infty, \infty}^{d}(\psi) = \bigcup_{k \in \na} \mathcal{NN}_{k, \infty}^{d}(\psi)$ is to be understood in an analogous way.
For the purpose of Example~\ref{ex:simple} below, let us also introduce the following notation: If $A$ is a finite set of functions $f \from \re \to \re$, then $\mathcal{NN}_{k,l}^{d}(A)$ denotes the set of feedforward neural networks, where the hidden nodes are endowed with either of the functions from $A$.
As a special case, we then have $\mathcal{NN}_{k,l}^{d}(A) = \mathcal{NN}_{k,l}^{d}(\psi)$ for $A = \{\psi\}$.

Functions in $\mathcal{NN}_{1, \infty}^{d}(\psi)$ are called shallow feedforward neural networks, while functions in $\mathcal{NN}_{\infty, \infty}^{d}(\psi)$ are generally referred to as deep feedforward neural networks.
The set $\mathcal{NN}_{\infty, l}^{d}(\psi) \subset \mathcal{NN}_{\infty, \infty}^{d}(\psi)$ of deep narrow networks, where $l \in \na$ is fixed, is also of special interest (cf. \cite{KL2020}).

In the context of Definition~\ref{def:fnn}, the function $\psi$ is sometimes also called squashing function, sigmoid function or ridge activation function.
Different terms have been chosen based on the properties of $\psi$, which in general differ based on which topological spaces we are studying the UAP in.
For the sake of simplicity, we call $\psi$ an activation function throughout this paper, and impose necessary properties on $\psi$ wherever needed.
In light of Lemma~\ref{lem:L2space4}, we are particularly interested in the UAP in $C(T; \re^{d})$.
At this point however, we need to discuss a technicality first.

For a given Borel measure $\nu$ on $T$ and  $f,g \in C(T; \re^{d})$, we write $f \sim_{\nu} g$ if $f = g$ outside a $\nu$-null set.
Note that $\sim_{\nu}$ is a binary relation on $C(T; \re^{d})$ which is reflexive, symmetric and transitive, hence $\sim_{\nu}$ is an equivalence relation.
We can thus consider the quotient space $C_{\nu}(T; \re^{d})$ of $C(T; \re^{d})$ under $\sim_{\nu}$, on which the $\nu$-essential supremum $\|\4\cdot\4\|_{L^{\infty}(T,\nu)}$ is a norm, making $(C_{\nu}(T; \re^{d}), \|\4\cdot\4\|_{L^{\infty}(T,\nu)})$ a normed vector space.
A modification of Lemma~\ref{lem:L2space4} shows that $C_{\nu}(T; \re^{d})$ is continuously embedded into $\Lambda^{2,0}$ as a dense linear subspace, provided that $\mu$ is absolutely continuous w.r.t. $\nu$.

Let us collect classical versions of the universal approximation theorem which are concerned with the (almost everywhere) uniform approximation of continuous functions, as we will be referring to them in the proofs of the subsequent results.

\begin{theorem}[{cf. \cite{H1991,KL2020,LLPS1993}}]\label{thm:UAT}
Given\/ $\psi \from \re \to \re$, consider the assumptions:
\begin{enumerate}[ref={\thetheorem(\arabic*)}, label=(\arabic*)]
\setlength\itemsep{0.1em}
\item\label{thm:UATAs1}
$\psi$ is continuous, bounded and non-constant;
\item\label{thm:UATAs2}
$\psi$ is continuous and nonaffine, and there exists a point\/ $x \in \re$ at which\/ $\psi$ is continuously differentiable with\/ $\psi'(x) \neq 0$;
\item\label{thm:UATAs3}
$\psi$ is locally\/ $\lambda$-essentially bounded.
Moreover,\/ $\psi$ is\/ $\lambda$-a.e. not an algebraic polynomial, and the set of points of discontinuity of\/ $\psi$ is a\/ $\lambda$-null set.
\end{enumerate}

Then:
\begin{enumerate}[ref={\thetheorem(\alph*)}, label=(\alph*)]
\setlength\itemsep{0.1em}
\item\label{thm:UAT1}
If \ref{thm:UATAs1} holds, then\/ $\mathcal{NN}_{1, \infty}^{d}(\psi)$ is dense in $C(T; \re^{d})$;
\item\label{thm:UAT2}
If \ref{thm:UATAs2} holds, then\/ $\mathcal{NN}_{\infty, d+3}^{d}(\psi)$ is dense in $C(T; \re^{d})$;
\item\label{thm:UAT3}
If \ref{thm:UATAs3} holds, then\/ $\mathcal{NN}_{1, \infty}^{d}(\psi)$ is dense in $C_{\lambda}(T; \re^{d})$.
\end{enumerate}
\end{theorem}

The following two Propositions~\ref{prop:NNdense} and~\ref{prop:NNdense2} yield dense subsets of $\Lambda^{2}$ which consist of feedforward neural networks.
We can therefore consider these sets as admissible for the set $D$ in the context of Convention~\ref{convention3}.
Consequently, due to Theorem~\ref{thm:CMapprox} and under suitable assumptions on $\psi$, feedforward neural networks are, up to the isometry $J$, dense in $H$.

Proposition~\ref{prop:NNdense} below is a consequence of Theorem~\ref{thm:UAT} and Lemma~\ref{lem:L2space4}.
For simplicity, we only formulate it for $\mathcal{NN}_{1, \infty}^{d}(\psi)$, where the case for $\mathcal{NN}_{\infty, d+3}^{d}(\psi)$ can be argued analogously.
Since $\mathcal{NN}_{1, \infty}^{d}(\psi)$ is a subset of $\mathcal{NN}_{k, \infty}^{d}(\psi)$ for every $k \in \na$, Propositions~\ref{prop:NNdense} and~\ref{prop:NNdense2} below do hold for $\mathcal{NN}_{k, \infty}^{d}(\psi)$,\/ $k \in \na$, too.

\begin{proposition}\label{prop:NNdense}
In the context of Theorem~\ref{thm:UAT}, assume either that Condition~\ref{thm:UATAs1} holds, or that Condition~\ref{thm:UATAs3} holds and\/ $\mu$ is absolutely continuous w.r.t.\/ $\lambda$.
Then\/ $\mathcal{NN}_{1, \infty}^{d}(\psi)$ is a dense linear subspace of\/ $\Lambda^{2,0}$.
\end{proposition}

By looking at the proof of Proposition~\ref{prop:NNdense} (which is presented in Appendix~\ref{app:technical}) it becomes clear that we cannot impose Assumption~\ref{thm:UATAs3} in case that $\mu$ is not absolutely continuous w.r.t. $\lambda$.
This is relevant in particular for Example~\ref{ex:pathological}, where we would need to impose either Assumption~\ref{thm:UATAs1} or Assumption~\ref{thm:UATAs2}.
Moreover, assuming $\psi$ to be ($\lambda$-a.e.) continuous is also rather restrictive, given that functions in $\Lambda^{2}$ need not be continuous.
By arguing along the lines of \cite{MR1015670,H1991}, we can actually drop the continuity assumption on $\psi$, at the cost of requiring boundedness, which is not required in Assumptions~\ref{thm:UATAs2} and~\ref{thm:UATAs3}.

\begin{proposition}\label{prop:NNdense2}
If\/ $\psi$ is bounded, measurable and nonconstant, then\/ $\mathcal{NN}_{1, \infty}^{d}(\psi)$ is a dense linear subspace of\/ $\Lambda^{2,0}$.
\end{proposition}

For notational simplicity, let us convene that the activation function $\psi$ satisfies sufficient conditions such that either Proposition~\ref{prop:NNdense} or~\ref{prop:NNdense2} is applicable.

\begin{convention}\label{convention4}
Henceforth we assume that either Condition~\ref{thm:UATAs1},~\ref{thm:UATAs3} or the assumptions from Proposition~\ref{prop:NNdense2} hold, depending on whether $\mu$ is absolutely continuous w.r.t. $\lambda$ or not, and whether we need to require $\psi$ to be ($\lambda$-a.e.) continuous.
\end{convention}

\begin{example}\label{ex:simple}
Fix $d = 1$ as well as $\mu = \lambda$.
Note that in this case $\pi \equiv 1$.
Set $D = \mathcal{NN}_{1, \infty}^{1}(\psi) = \mathrm{span}\bigl\{ T \ni t \mapsto \psi(\alpha t + \eta)\ |\ \alpha, \eta \in \re \bigr\}$, where $\psi = \tanh$.
Since every $f \in D$ is continuous and thus bounded on $T$, we may replace the Lebesgue- by the Riemann integral.

If $\alpha = 0$ and $\eta \in \re$, then $\int_{0}^{t} \psi(\eta)\, \mathrm{d}s = \psi(\eta) t$ for $t \in T$.
On the other hand, if $\alpha \neq 0$ and $\eta \in \re$ then, by substitution,
\begin{equation*}
\int_{0}^{t} \psi(\alpha s + \eta)\, \mathrm{d}s = \frac{1}{\alpha} \bigl( \tilde{\psi}(\alpha t + \eta) - \tilde{\psi}(\eta) \bigr),
\quad
t \in T,
\end{equation*}
where $\tilde{\psi}(\cdot) = \log(\cosh(\cdot))$.
Similarly, if $\psi$ is the standard sigmoid (logistic) function, then the same applies with $\tilde{\psi}(\cdot) = \log(1 + \exp(\cdot))$, which is also called softplus function.
To sum up, we see that
\begin{equation}\label{eq:HDchar}
H(D) = \mathrm{span}\bigl\{ \mathrm{id} \colon T \ni t \mapsto t,\, \mathcal{NN}_{1, \infty}^{1}(\tilde{\psi}) \bigr\} = \mathcal{NN}_{1, \infty}^{1}(\{\mathrm{id}, \tilde{\psi}\}).
\end{equation}
\end{example}

Provided that $\psi$ is Riemann integrable, Example~\ref{ex:simple} shows that, compared to the set $D = \mathcal{NN}_{1, \infty}^{1}(\psi)$, the set $H(D)$ can be obtained by modifying the activation function, and adding the identity function into the set of admissible activation functions.
This can be helpful when optimizing over functions in $H(D)$, because one avoids having to implement integral operations.
What is more, functions in $H(D)$ enjoy the property of being absolutely continuous, provided that $\mu$ is absolutely continuous w.r.t. $\lambda$, while this is not always the case for functions from $\mathcal{NN}_{1, \infty}^{1}(\psi)$, e.g. if $\psi$ is not continuous.

Note that we formulated Propositions~\ref{prop:NNdense} and~\ref{prop:NNdense2} for shallow feedforward neural networks.
However, as already mentioned above, Propositions~\ref{prop:NNdense} and~\ref{prop:NNdense2} do hold for the set of deep neural networks, too.
For a discussion on the topic of depth vs. width, see for example \cite{LPWHW2017,RO2016}.

\subsection{Interlude: Universal approximation in H{\"o}lder norm}\label{subsec:UATHoelder}

In Theorem~\ref{thm:UAT}, we cited classical versions of the universal approximation theorem which are concerned with (almost everywhere) uniform approximation of continuous functions.
Note that this is, in essence, a topological statement, and we may seek for refined approximation results that hold w.r.t. stricter topologies.
Natural candidate topologies w.r.t. which we may seek to derive a universal approximation theorem are H{\"o}lder type topologies.

Given $\alpha \in (0, 1)$, we denote by $E_{\alpha} = C_{0}^{\alpha}(T; \re^{d}) \subset C_{0}(T; \re^{d})$ the vector space of $\re^{d}$-valued, $\alpha$-H{\"o}lder continuous functions on $T$ that are zero at the origin.
The space $E_{\alpha}$ is also referred to as $\alpha$-H{\"o}lder space.
We endow this space with the topology which is induced by the norm
\begin{equation*}
E_{\alpha} \ni f \mapsto \|f\|_{\alpha} \coloneqq \sup_{\substack{s,t \in T \\ 0 < t-s \le 1}} \frac{|f(t) - f(s)|}{(t-s)^{\alpha}}.
\end{equation*}

Kolmogorov--Chentsov's continuity theorem shows that all paths of $\re^{d}$-valued standard Brownian motion are $\alpha$-H{\"o}lder continuous for every $\alpha \in (0, 1/2)$, hence it would be desirable to use $E_{\alpha}$ as the space on which to consider the restriction of the classical Wiener measure (see Example~\ref{ex:CWS} and Definition~\ref{app:GM}).
Although $(E_{\alpha}, \|\4\cdot\4\|_{\alpha})$ is indeed a real Banach space, it does not contain a countable dense subset and is thus not separable (cf. \cite[Exercise~2.116]{Schmock2021}).

A solution to this issue is to pass to the little $\alpha$-H{\"o}lder space, i.e. the subspace $E_{\alpha, 0}$ of all $f \in E_{\alpha}$ that satisfy $|f(t)-f(s)| = o(|t-s|^{\alpha})$ as $|t-s| \searrow 0$.
Then $(E_{\alpha, 0}, \|\4\cdot\4\|_{\alpha})$ is a real Banach space.
$E_{\alpha, 0}$ is also referred to as the space of $\alpha$-H{\"o}lder paths with vanishing H{\"o}lder oscillation (cf. \cite[Exercise~2.12]{MR4174393}).
Note that $E_{\alpha, 0}$ has a very useful characterization: It is the closure of $C_{0}^{\infty}(T; \re^{d})$, the vector space of $\re^{d}$-valued smooth functions on $T$ that are zero at the origin, where the closure is taken with respect to the topology induced by $\|\4\cdot\4\|_{\alpha}$.
Moreover, we have the inclusion $E_{\beta} \subset E_{\alpha, 0}$ for all $0 < \alpha < \beta < 1$.

In the context of Gaussian measures and large deviations theory, the space $E_{\alpha, 0}$ has been studied in great detail (cf. \cite{MR3070442,MR1172514,MR132389}).
In particular, the following important property has been shown to hold: If we fix $\pi = I_{d}$ and $\mu = \lambda$, then $H$ is continuously embedded into $E_{\alpha,0}$ for each $\alpha \in (0, 1/2)$, and there exists a countable family of functions in $H$ (the Faber--Schauder system) that constitutes a Schauder basis of $(E_{\alpha, 0}, \|\4\cdot\4\|_{\alpha})$.
While on the one hand this implies the separability of $E_{\alpha, 0}$, more importantly, we see that $H$ is not only continuously, but also densely embedded into $E_{\alpha, 0}$.
In conjunction with Theorem~\ref{thm:CMapprox}, and as a direct consequence to this observation, we obtain the following result.

\begin{proposition}[Universal approximation in H{\"o}lder norm]\label{prop:UAThoelder}
Fix\/ $\mu = \lambda$,\/ $\pi = I_{d}$ and\/ $\alpha, \beta \in (0, 1/2)$ with\/ $\beta < \alpha$.
Then, in the context of Propositions~\ref{prop:NNdense} and~\ref{prop:NNdense2}, for each\/ $f \in E_{\alpha, 0}$, there exists a sequence\/ $(f_{n})_{n \in \na}$ in\/ $\mathcal{NN}_{1, \infty}^{d}(\psi)$ such that
\begin{equation*}
\lim_{n \to \infty} \|f - J(f_{n})\|_{\alpha} = 0.
\end{equation*}
In particular, every smooth function\/ $f \in C_{0}^{\infty}(T; \re^{d})$ and every\/ $\alpha$-H{\"o}lder continuous function\/ $f \in E_{\alpha}$ can be approximated, up to the linear isometry\/ $J$ (see Definition~\ref{def:CMspace}), by sequences from\/ $\mathcal{NN}_{1, \infty}^{d}(\psi)$ w.r.t.\/ $\|\4\cdot\4\|_{\alpha}$ and\/ $\|\4\cdot\4\|_{\beta}$, respectively.
\end{proposition}

Let us conclude this subsection with several remarks.
First, note that, based on \cite{MR1244574}, Proposition~\ref{prop:UAThoelder} should extend to certain Besov--Orlicz type norms, which induce stricter topologies than the H{\"o}lder norms.
Moreover, it should be possible to relax the assumption $\mu = \lambda$ by considering a modified H{\"o}lder norm with denominator $\mu\bigl((s,t]\bigr)^{\alpha}$ instead of $(t-s)^{\alpha}$.
Finally, note that the universal approximation property of neural networks in topological spaces with topologies that are stricter than the one induced by the uniform norm have already been studied in the literature. See e.g. \cite{MR4131039}, which studies the UAP in Sobolev spaces.

\section{Importance sampling with feedforward neural networks}\label{sec:IS}

Upon having studied the tractable space $H$ of drift adjustments which coincides with the Cameron--Martin space of the Gaussian measure $\gamma_{M}$, and having proved that feedforward neural networks are, up to the isometry $J$, dense in $H$ by combining Theorem~\ref{thm:CMapprox} and~Propositions \ref{prop:NNdense} resp.~\ref{prop:NNdense2}, we now turn our attention to importance sampling.
To this end, we first write down the basic setting.
In Subsection~\ref{subsec:UATGR}, we then study how our method complements a classical approach which employs ideas from large deviations theory, before finally studying the full problem in Subsection~\ref{subsec:densityapprox}.
Most notably, Theorem~\ref{thm:densityapprox} below provides a theoretical justification for our simulations in Section~\ref{sec:simulation}, see also Remark~\ref{rem:nndense}.

Let $\mathcal{C}$ denote the vector space of $\re^{n}$-valued, continuous and $\FF$-adapted processes.
Let $a \from \Omega \times T \times \mathcal{C} \to \re^{n}$ and $b \from \Omega \times T \times \mathcal{C} \to \re^{n\times d}$ be non-anticipative coefficients (cf. \cite[Definition~16.0.3]{MR3443368}), such that the stochastic differential equation
\begin{equation}\label{eq:priceSDE}
\mathrm{d}X_{t} = a_{t}(X)\, \mathrm{d}C_{t} + b_{t}(X)\, \mathrm{d}M_{t},
\quad t \in T,
\end{equation}
subject to $X_{0} \equiv x \in \re^{n}$, admits a unique weak solution.
For ease of notation, we write $a_{t}(X), b_{t}(X)$ instead of $a(\omega, t, X), b(\omega, t, X)$, and understand Equation~\eqref{eq:priceSDE} to hold componentwise, i.e. $X_{t}^{i} = x^{i} + \bigl(a(X)_{i} \sbullet[.75] C\bigr)_{t} + \bigl(b(X)_{i,\cdot} \sbullet[.75] M\bigr)_{t}$ for each $i \in \{1,2,\hdots,n\}$ and $t \in T$.

Let $F \from \Omega \times C(T; \re^{n}) \to \re$ be a random functional, such that the mapping $\Omega \ni \omega \mapsto F(\omega, X(\omega))$ is $\Fcal_{u}$-measurable.
For simplicity we write $F(\cdot, X) = F(X)$, and call $F(X)$ a random payoff.
We are interested in obtaining a Monte Carlo estimate of its expectation under $\PP$,
\begin{equation}\label{eq:MC}
\mathbb{E}_{\PP}[F(X)] = \int_{\Omega} F(\omega, X(\omega))\, \PP(\mathrm{d}\omega),
\end{equation}
provided that \eqref{eq:MC} is real-valued.

\begin{remark}
If $\mathbb{E}_{\PP}[F(X)]$ is to denote an option price, then we would require $\PP$ to be a risk-neutral measure.
However, the results in this section do not require $\PP$ to be risk-neutral.
Actually, we do not need to assume $\mathbb{E}_{\PP}[F(X)]$ to be an option price, as long as $X$ follows the SDE \eqref{eq:priceSDE} and $F(X) \in \mathcal{L}^{0}(\Fcal_{u})$.
\end{remark}

\begin{example}[Stochastic volatility models]\label{ex:stochvol}
The SDE \eqref{eq:priceSDE} can model the evolution of asset prices within stochastic volatility models.
Thus, our method complements \cite{MR2565852}, where the author employs methods from the theory of large deviations in order to derive asymptotically optimal drift adjustments (see Subsection~\ref{subsec:UATGR}, where we discuss asymptotic optimality in the context of \cite{MR2362149}) for pricing stochastic volatility models, very much in the spirit of \cite{MR2362149}.
\end{example}

Let $(X_{i})_{i \in \na}$ denote a sequence of independent copies of solutions to \eqref{eq:priceSDE}.
By the strong law of large numbers, the sample means $Z_{k} = \sum_{i=1}^{k} F(X_{i}) / k$ converge $\PP$-almost surely to $m = \mathbb{E}_{\PP}[F(X)]$.
Moreover, if $F(X_{1})$ has a finite variance $\sigma^{2} > 0$ then, according to the central limit theorem, as $k \to \infty$, the law of $\sqrt{k}(Z_{k} - m)$ converges weakly to $\N(0, \sigma^{2})$.
We therefore see that $Z_{k} - m$ is approximately normally distributed with mean zero and standard deviation $\sigma / \sqrt{k}$.
In practice, the standard error $\sigma / \sqrt{k}$ can be quite large even for large values of $k$, which calls for the application of variance reduction methods.

\begin{remark}\label{rem:intsim}
For each $f \in \Lambda^{2}$, we have that $f^{\top} \sbullet[.75]M$ is a real-valued continuous local martingale with $[f^{\top} \sbullet[.75] M]_{t} = \int_{[0,t]}f^{\top}(s)\pi(s)f(s)\, \mu(\mathrm{d}s)$ for $t \in T$.
As a consequence, similarly as argued in Remark~\ref{rem:Gprop}, $f^{\top} \sbullet[.75]M$ is a Gaussian $\FF$-Markov process with
\begin{equation*}
\mathcal{L}\Bigl( (f^{\top} \sbullet[.75]M)_{t} - (f^{\top} \sbullet[.75]M)_{s} \Bigr) = \N\Bigl(0, \int_{(s,t]}f^{\top}(u)\pi(s)f(u)\, \mu(\mathrm{d}u)\Bigr),
\end{equation*}
for $s < t$ in $T$, which shows how one can simulate increments of $f^{\top} \sbullet[.75]M$, provided that the integrals $\int_{(s,t]}f^{\top}(u)\pi(s)f(u)\, \mu(\mathrm{d}u)$ can be explicitly computed.
\end{remark}

Recall that, whenever $Y$ is a real-valued continuous semimartingale, the Dol\'{e}ans exponential $\mathcal{E}(Y)$ is the strictly positive continuous semimartingale that is, up to indistinguishability, the unique solution in $L(Y)$ to the stochastic integral equation
\begin{equation*}
\mathcal{E}(Y) = \exp{(Y_{0})} + \int_{0}^{\cdot}\mathcal{E}(Y)_{s}\, \mathrm{d}Y_{s},
\end{equation*}
and is given by $\mathcal{E}(Y) = \exp{(Y - [Y]/2)}$ (cf. \cite[Theorem~6.47]{Schmock2021}).
By a generalization of the functional equation of the exponential function, we further have $\mathcal{E}(Y)^{-1} = \mathcal{E}(-Y + [Y])$.
We refer to~\cite{R2010} for a survey on stochastic exponentials.

Since $M$ is a continuous local martingale and $[f_{h}^{\top} \sbullet[.75] M]_{u} = \| h \|_{H}^{2} < \infty$ for each $h \in H$, we know that $f_{h}^{\top}\sbullet[.75] M$ is a square integrable martingale, i.e. $f_{h}^{\top}\sbullet[.75] M \in \mathcal{H}^{2}$ (see also Lemma~\ref{lem:IKWapprox} in Appendix~\ref{app:technical}), and $\mathcal{E}(f_{h}^{\top} \sbullet[.75] M)$ is a non-negative continuous local martingale, hence a supermartingale.
Moreover, since $\mathbb{E}_{\PP}\bigl[\exp{\bigl( \tfrac{k}{2} [f_{h}^{\top}\sbullet[.75]M]_{u} \bigr)}\bigr] = \exp{\bigl( \tfrac{k}{2}\|h\|_{H}^{2}\bigr)} < \infty$ for every $k > 1$,
\cite[Theorem~15.4.6]{MR3443368} and \cite{MR602524} show that $\mathcal{E}(f_{h}^{\top} \sbullet[.75] M) \in \mathcal{H}^{p}$ for every $p = k/(2\sqrt{k}-1) > 1$ with upper bound
\begin{equation*}
\| \mathcal{E}(f_{h}^{\top} \sbullet[.75] M) \|_{\mathcal{H}^{p}} \le \tfrac{p}{p-1} \exp{\bigl( \tfrac{\sqrt{k}-1}{2}\|h\|_{H}^{2}\bigr)}.
\end{equation*}
An application of de la Vall\'{e}e Poussin's criterion (see e.g. \cite[Corollary 2.5.5]{MR3443368}) further shows that $\mathcal{E}(f_{h}^{\top} \sbullet[.75] M) \in \mathcal{H}^{p}$ implies the uniform integrability of $\mathcal{E}(f_{h}^{\top} \sbullet[.75] M)$.

By a change of measure, Equation~\eqref{eq:MC} can now be rewritten as
\begin{equation}\label{eq:MCdrift}
\mathbb{E}_{\PP}[F(X)] = \mathbb{E}_{\PP_{h}}[F(X) (\mathcal{E}(f_{h}^{\top}\sbullet[.75]M)^{-1})_{u}],
\end{equation}
where $\PP_{h}$ is defined by $\mathrm{d}\PP_{h} = \mathcal{E}(f_{h}^{\top}\sbullet[.75]M)_{t}\, \mathrm{d}\PP$ on $\Fcal_{t}$ for every $t \in T$.
If $F_{h}(X) \coloneqq F(X)(\mathcal{E}(f_{h}^{\top}\sbullet[.75]M)^{-1})_{u}$ denotes the modified random payoff, then the $\PP$-expectation of $F(X)$ and the $\PP_{h}$-expectation of $F_{h}(X)$ are identical.
Provided that $F(X)$ has a finite second moment w.r.t. $\PP$, the variance of $F_{h}(X)$ under $\PP_{h}$ is given by
\begin{equation}\label{varianceQ}
\mathbb{E}_{\PP_{h}}[F_{h}^{2}(X)] - \mathbb{E}_{\PP_{h}}[F_{h}(X)]^{2} = \mathbb{E}_{\PP}[F^{2}(X)(\mathcal{E}(f_{h}^{\top}\sbullet[.75]M)^{-1})_{u}] - \mathbb{E}_{\PP}[F(X)]^{2}.
\end{equation}
Therefore, we can compute~\eqref{eq:MCdrift} under the measure $\PP_{h}$ and try to find $h \in H$ such that~\eqref{varianceQ} is minimized.
Note that the second term on the right-hand side of~\eqref{varianceQ} does not depend on $h$, so that we focus on minimizing the first term, which for each $h \in H$ is given by
\begin{equation*}
V(h) \coloneqq \mathbb{E}_{\PP}[F^{2}(X)(\mathcal{E}(f_{h}^{\top}\sbullet[.75]M)^{-1})_{u}] = \mathbb{E}_{\PP}\bigl[ F^{2}(X)\exp\bigl(-(f_{h}^{\top}\sbullet[.75]M)_{u}+\|h\|_{H}^{2}/2\bigr)\bigr].
\end{equation*}

To sum up, our problem reads
\begin{equation}\label{goal}
\min_{h \in H} V(h),
\end{equation}
provided that a minimizer of $V$ exists (see Theorem~\ref{thm:densityapprox} for sufficient conditions).

\subsection{Approximating the asymptotically optimal sampling measure}\label{subsec:UATGR}

Before we turn our attention to solving \eqref{goal}, let us first discuss the classical approach presented in \cite{MR2362149}, that uses methods from the theory of large deviations, and show how our method complements it.
Note that the setting of \cite{MR2362149} is a special case of the setting of Section~\ref{sec:IS}.
To this end, set $d = 1$, let $M$ be a standard Brownian motion and let $\FF$ be the augmented natural filtration of $M$ (which satisfies the usual hypotheses).
Set $X = M$ and assume that the payoff $F \from C_{0}(T; \re) \to \replus$ is continuous, where $C_{0}(T; \re)$ is endowed with the topology of uniform convergence.

In \cite{MR2362149}, the authors argue that \eqref{goal} is in general intractable.
Rather than minimizing \eqref{goal}, the authors consider for each $h \in H$ the small-noise limit
\begin{equation}\label{goallimit}
L(h) \coloneqq \limsup_{\epsilon \searrow 0} \epsilon \log{\mathbb{E}_{\PP}\Bigl[ \exp\Bigl(\tfrac{1}{\epsilon}\bigl(2\tilde{F}(\sqrt{\epsilon}M)-\bigl((\sqrt{\epsilon}f_{h})\sbullet[.75]M\bigr))_{u}+\|h\|_{H}^{2}/2 \bigr)\Bigr)}\Bigr],
\end{equation}
where $\tilde{F} \coloneqq \log{F}$.
The limit \eqref{goallimit} corresponds to approximating $V(h) \approx \exp(L(h))$.

Assume that $\tilde{F} \from C_{0}(T; \re) \to \re \cup \{-\infty\}$ is continuous, and that there exist constants $K_{1}, K_{2} > 0$ as well as $\alpha \in (0,2)$ such that $\tilde{F}(x) \le K_{1} + K_{2} \|x\|_{\infty}^{\alpha}$ for each $x \in C_{0}(T; \re)$.
As \cite[Theorem~3.6]{MR2362149} shows, one can invoke a version of Varadhan's integral lemma to rewrite \eqref{goallimit} as a variational problem, provided that $h$ is an element of  $H_{\mathrm{bv}}$, the space of all $h \in H$ such that $f_{h} \in \Lambda^{2} = L^{2}(\lambda)$ is of bounded variation, and aim to solve $\min_{h \in H_{\mathrm{bv}}} L(h)$, provided that a minimizer exists.

For the purpose of the proof of the central result \cite[Theorem~3.6]{MR2362149}, the following functional is of importance: For $M > 0$ and $h \in H$ consider $\tilde{F}_{h,M} \colon H \ni g \mapsto 2\tilde{F}(g) - M\|g+h\|_{H}^{2} + \|h\|_{H}^{2}$.
By \cite[Lemma~7.1]{MR2362149}, there exists a maximizer $g_{h,M} \in H$ of $\tilde{F}_{h,M}$.
Together with Proposition~\ref{prop:NNdense}, we then obtain

\begin{proposition}\label{prop:asoptapprox}
Assume that\/ $\psi \from \re \to \re$ is continuously differentiable, bounded and non-constant, and set\/ $D = \mathcal{NN}_{1,\infty}^{1}(\psi)$.
Then\/ $H(D) \subset H_{\mathrm{bv}}$, and for each\/ $M > 0$ and\/ $h \in H$, there exists a sequence\/ $(h_{n})_{n \in \na}$ in\/ $H(D)$ such that
\begin{equation*}
\lim_{n \to \infty} \tilde{F}_{h,M}(h_{n}) = \tilde{F}_{h,M}(g_{h,M}) = \max_{g \in H} \bigl(2\tilde{F}(g) - M\|g+h\|_{H}^{2} + \|h\|_{H}^{2}\bigr).
\end{equation*}
\end{proposition}

According to \cite[Theorem~3.6]{MR2362149} and the discussion thereafter, the general strategy for finding a minimizer of $L$ is as follows:
Find a maximizer $g_{h,1}$ to $\tilde{F}_{h, 1}$ for $h \equiv 0$.
Verify whether $g_{h,1}$ is actually an element of $H_{\mathrm{bv}}$ and, if this is the case, then $g_{h,1}$ minimizes $L$ provided that $L(g_{h,1}) = \tilde{F}_{h, 1}(g_{h,1})$ holds true.
Since the evaluation of $L$ at $g_{h,1}$ involves having to find a maximizer of $\tilde{F}_{g_{h,1},1/2}$, checking the identity $L(g_{h,1}) = \tilde{F}_{h, 1}(g_{h,1})$ might only be feasible by a numerical approximation which introduces an error.
However, if one could establish said identity, then $g_{h,1}$ would be a minimizer of $L$, in which case we say that $g_{h,1}$ is asymptotically optimal.

Let us assume that there exists a maximizer $g_{h,1}$ to $\tilde{F}_{h, 1}$ for $h \equiv 0$ that is indeed asymptotically optimal.
In light of Proposition~\ref{prop:asoptapprox}, we can then approximate $g_{h,1}$ by a sequence $(h_{n})_{n \in \na}$ from $H(\mathcal{NN}_{1,\infty}^{1}(\psi))$ such that $\tilde{F}_{h,1}(h_{n})$ converges to $\tilde{F}_{h,1}(g_{h,1})$.
Theorem~\ref{thm:densityapprox2} below then implies that $V(h_{n})$ converges to $V(g_{h,1})$.
To sum up, rather then trying to find a minimizer of $V$, one might instead study
\begin{equation}\label{goalmodified}
\max_{g \in H} \tilde{F}_{h,1}(g) = \max_{g \in H} \bigl( 2 \tilde{F}(g) - \|g\|_{H}^{2} \bigr)
\end{equation}
and solve the modified problem \eqref{goalmodified} with feedforward neural networks.

\subsection{Approximating the optimal sampling measure}\label{subsec:densityapprox}

In what follows, we consider the full problem~\eqref{goal} and propose to solve it with feedforward neural networks.
One advantage of this approach is that we do not, in contrast to Subsection~\ref{subsec:UATGR}, seek to find an asymptotically optimal drift adjustment by minimizing \eqref{goallimit}, but rather stay within the full problem~\eqref{goal}.
Moreover, Theorem~\ref{thm:densityapprox4} and Remark~\ref{rem:nndense} provide a theoretical justification for employing the tractable class of shallow feedforward neural networks for this optimization problem.
The numerical simulations in Section~\ref{sec:simulation} will demonstrate that indeed, we obtain substantial reductions in the variance of the Monte Carlo estimators for several multivariate asset price processes and path-dependent payoff functionals.

In the following lemma, we consider a nonlinear operator which maps elements from Cameron--Martin space to probability densities.
This result is essential, as it will imply in Theorem~\ref{thm:densityapprox} below that the optimal sampling measure can be approximated by measures which are generated by feedforward neural networks.

\begin{lemma}\label{lem:density}
The operator\/ $A_{p} \colon H \ni h \mapsto (\mathcal{E}(f_{h}^{\top}\sbullet[.75]M)^{-1})_{u} \in L^{p}(\PP)$ is continuous for each\/ $p \in [1, \infty)$.
Moreover,\/ $A_{p}$ is not quasi-bounded, meaning that
\begin{equation*}
\limsup_{\|h\| \to \infty} \frac{\|A_{p}(h)\|_{L^{p}(\PP)}}{\|h\|_{H}} = \infty.
\end{equation*}
\end{lemma}

Finally, we formulate Theorem~\ref{thm:densityapprox}.
The proof (which is presented in Appendix~\ref{app:technical}) complements \cite[Proposition~4]{MR2680557} and not only shows, under rather weak assumptions, that the functional $V$ does indeed admit a minimizer, but Theorem~\ref{thm:densityapprox4} applied to a dense subset $D$ of $\Lambda^{2}$ which consists of feedforward neural networks provides the theoretical justification for the simulations in Section~\ref{sec:simulation}, see also Remark~\ref{rem:nndense} below.

\begin{theorem}\label{thm:densityapprox}
Assume that\/ $\PP\bigl(\{ F^{2}(X) > 0 \} \bigr) > 0$, and that there exists some\/ $\varepsilon > 0$ such that\/ $F(X) \in L^{2+\epsilon}(\PP)$.
Then:

\begin{enumerate}[ref={\thetheorem(\alph*)}, label=(\alph*)]
\setlength\itemsep{0.1em}
\item\label{thm:densityapprox1}
$V$ is\/ $\replus$-valued;
\item\label{thm:densityapprox2}
$V$ is continuous;
\item\label{thm:densityapprox3}
There exists a minimizer of\/ $V$, i.e.
\begin{align*}
\argmin_{h \in H}V(h) = \bigl\{ g \in H\ |\ V(g) \le V(h),\ \forall\, h \in H \bigr\} \neq \emptyset;
\end{align*}
\item\label{thm:densityapprox4}
There exists a sequence\/ $(h_{n})_{n \in \na}$ in\/ $H(D)$ such that\/ $\lim_{n \to \infty} V(h_{n}) = \min_{h \in H}V(h)$.
\end{enumerate}
\end{theorem}

\begin{remark}\label{rem:nndense}
In the context of Theorem~\ref{thm:densityapprox4}, we may seek to find a minimizer of $V$ by performing measure changes which are induced by Dol\'{e}ans exponentials of the form $\mathcal{E}(f^{\top}\sbullet[.75]M)$, where $f \in \mathcal{NN}_{1, \infty}^{d}(\psi)$.
In Section~\ref{sec:simulation}, we pursue this approach for several different asset price models, achieving substantial reductions in the variance of the corresponding Monte Carlo estimators.
\end{remark}

\begin{remark}
Theorem~\ref{thm:densityapprox4} shows that neural network-induced changes of the sampling measure can approximate the optimal sampling measure arbitrarily well in the sense that the second moment of the modified payoff under the optimal measure can be approximated up to an arbitrarily small $\epsilon > 0$. 
However, the proof is not constructive, it does not deliver a recipe how to actually obtain such a sequence $(h_{n})$ of neural network-induced elements from Cameron--Martin space that converges to the optimum.
In Section 5 below, we use stochastic gradient descent to train our neural networks.
This procedure builds on the method of stochastic approximation, which was pioneered by Robbins \& Monro in 1951 (\cite{MR42668}).
Stochastic approximation for importance sampling for option pricing in continuous-time models has been studied by Lemaire \& Pag\`{e}s (\cite{MR2680557}).
We refer to their Section 3 for details on how to construct convergent sequences of functions based on the method of stochastic approximation.
\end{remark}

\begin{remark}[Extension to the calculation of sensitivities]\label{rem:sen}
Let us assume that the SDE~\eqref{eq:priceSDE} depends on a set of parameters $\alpha \in \re^{m}$ for some $m \in \na$.
Fix $i \in \{1,2,\hdots,m\}$, and let us further assume that we can exchange the order of differentiation and integration, i.e. $\tfrac{\partial}{\partial \alpha_{i}} \mathbb{E}_{\PP}[F(X)] = \mathbb{E}_{\PP}[\tfrac{\partial}{\partial \alpha_{i}} F(X)]$.
If we wanted to jointly reduce the standard error of the Monte Carlo estimators of the expected random payoff and of its sensitivity w.r.t. $\alpha_{i}$, we could modify the definition of $V$:
\begin{equation*}
\tilde{V}(h) = \mathbb{E}_{\PP}\bigl[\bigl(w_{1} F^{2}(X) + w_{2}\bigl(\tfrac{\partial}{\partial \alpha_{i}} F(X)\bigr)^{2}\bigr)(\mathcal{E}(f_{h}^{\top}\sbullet[.75]M)^{-1})_{u}\bigr],
\quad h \in H,
\end{equation*}
where $w_{1}, w_{2} \in (0,1)$ are weights that sum up to $1$.
If there exists some\/ $\varepsilon > 0$ such that $w_{1}F^{2}(X) + w_{2}(\tfrac{\partial}{\partial \alpha_{i}} F(X))^{2} \in L^{1+\varepsilon}(\PP)$ and $\PP\bigl(\{ w_{1}F^{2}(X) + w_{2}(\tfrac{\partial}{\partial \alpha_{i}} F(X))^{2} > 0 \} \bigr) > 0$, then Theorem~\ref{thm:densityapprox} applies correspondingly.
Analogous considerations hold for higher-order sensitivities as well as for the joint reduction of standard errors for more than one sensitivity.
We refer the reader to~\cite[Section 7.2]{MR1999614} for details on the computation of pathwise derivatives for some classical models and payoffs.
\end{remark}

\section{Numerical study}\label{sec:simulation}

In this section, we provide a range of carefully chosen numerical examples to showcase the various strengths of our method.
Additionally, we will compare our approach to other methods that have been proposed in the literature.
All computational tasks were performed using Python, leveraging the Keras deep learning API for the construction and training of our feedforward neural networks.\footnote{\label{fn:Github}All codes that were used for the simulations are available on Github, see \href{https://github.com/aarandjel/importance-sampling-with-feedforward-neural-networks}{https://github.com/aarandjel/importance-sampling-with-feedforward-neural-networks}.}

Let us provide a brief overview of the examples appearing in the subsequent subsections.
In Subsection~\ref{subsec:business}, we explore a time-change instance that deviates from the conventional assumption of $\mu = \lambda$ to better represent phases of changing business activity.
Subsection~\ref{subsec:knockout} considers a knock-out option and discusses the occurrence of multiple rare events.
Moving on to Subsection~\ref{subsec:dyncor}, we examine a stochastic volatility model with an imposed dynamic correlation structure, which directly influences the norm on Cameron--Martin space.
Lastly, in Subsection~\ref{subsec:basket}, we investigate the feasibility of utilizing neural networks for importance sampling in a high-dimensional model.
Throughout all of our examples, we consider arithmetic Asian (basket) call options with strike $K$ and basket weights $w$ as the chosen payoffs,
\begin{align}\label{eq:payoff}
F(X) = \Bigl( \frac{1}{u}\int_{0}^{u} \langle w, X_{t} \rangle\, \mathrm{d}t - K \Bigr)^{+},
\end{align}
while Subsection~\ref{subsec:knockout} additionally incorporates knock-out barriers for further analysis.

To establish a solid basis for comparison, we have selected the methodologies proposed by Glassermann et al. (\cite{MR1849001}), Guasoni \& Robertson (\cite{MR2362149}), Capriotti (\cite{MR2451625}), Arouna (\cite{A2004}), Su \& Fu (\cite{899767}), as well as Jourdain \& Lelong (\cite{MR2680557}).
To underscore the versatility of our approach in handling more general models than those presented in the literature, we will initially present results for the models discussed in the previous paragraph.
Subsequently, we will report results from simulations performed for the models studied in the literature mentioned above.

To train a feedforward neural network, our approach is as follows.
First, we simulate a fixed number $N$ of trajectories $X^{i}$ of the asset price using the Euler--Maruyama method, based on a pre-defined time-grid.
Then, we decide on a set $\mathcal{NN}_{k,l}^{d}(\psi)$ from which we seek to identify the optimal function, by selecting the number of hidden layers $k$, the number of hidden nodes $l$, and the activation function $\psi$.
The output dimension $d$ of the neural networks aligns with the dimension of the process $M$.
We approximate $V$ by computing an average over the $N$ trajectories,
\begin{equation}\label{eq:SGDgoal}
V(\theta) = \frac{1}{N} \sum_{i=1}^{N} F^{2}(X^{i})\exp\bigl(-(f_{\theta}^{\top}\sbullet[.75]M^{i})_{u}+\|f_{\theta}\|_{\Lambda^{2}}^{2}/2\bigr),
\end{equation}
where $\theta$ represents the vector encompassing all trainable parameters of the neural network $f_{\theta}$, and all quantities on the right-hand side of Equation~\eqref{eq:SGDgoal} are appropriately discretized.
We therefore consider $V$ as a function of the finite parameter vector $\theta$, and aim to find the optimal $\theta^{*}$ and thus the optimal element $f_{\theta^{*}}$ from $\mathcal{NN}_{k,l}^{d}(\psi)$.

To achieve this, we employ stochastic gradient descent, a technique originally pioneered by Robbins \& Monro (\cite{MR42668}).
Specifically, we adopt the mini-batch variant of this method, which replaces the mean over all $N$ trajectories with means over smaller sub-batches.
Starting from an initial guess, the parameter-vector $\theta$ is then iteratively updated with a scaled version of the gradient of $V$ over those sub-batches, i.e. $\theta_{m+1} = \theta_{m} - \gamma_{m} \nabla_{batch}V(\theta_{m})$ with learning rate $\gamma_{m}$ and $\nabla_{batch}V$ denoting the gradient of $V$ over one specific batch.
Upon completing a full iteration through all batches, we consider the neural network to have completed one epoch of training.
For each subsequent epoch, the trajectories contained in the individual batches can then be randomly shuffled around, and the parameter $\theta$ is updated until a stopping criterion is reached.
One notable advantage of neural networks lies in their ability to efficiently compute gradients through the back-propagation method.
Additionally, we utilize a popular modified version of this training routine known as Adam (\cite{KB2017}), which incorporates the first and second moments of the gradient estimates to enhance performance.

In all of our subsequent examples, we train the neural networks using 100 batches, each consisting of 1,024 trajectories.
For validation purposes, we employ an additional 100 batches, also comprising 1,024 trajectories, and stop the training process when the loss, $V(\theta)$, ceases to reduce on the validation set.
The results presented in the following tables are derived from simulations performed on separate test datasets, each containing $10^5$ trajectories.
Throughout the training, validation, and testing phases, we maintain a fixed learning rate of $10^{-3}$ for the stochastic gradient descent, and we fix the time horizon to $u=1$ to consider the interval $T=[\4 0, 1]$.
Unless otherwise specified, we utilize a step size of $\Delta t = 1/250$.
However, in Subsection~\ref{subsec:basket}, we deviate from this convention. 
We employ a step size corresponding to $\Delta t = d / 10^4$ during the training and validation process, where $d$ denotes the dimension of the asset price process.
For example, when $d = 200$ then $\Delta t = 1/50$.
This adjustment is only implemented for dimensions ranging from $d=100$ to $d=1{,}000$, while the step size always remains $\Delta t = 1/100$ for the testing dataset, as well as for the training and validation datasets in case $d < 100$.
In all simulations described below, we train shallow feedforward neural networks with a single hidden layer, using $\psi(x) = \tanh(x)$ as activation function.
The number of hidden nodes used for the various examples is reported beneath the tables.
The tables below present results for different choices of model parameters, presenting mean estimates, standard errors, relative standard errors as a percentage of the mean, and variance ratios.
The variance ratios were obtained by comparing the variance of the mean estimate from both a Monte Carlo and a Monte Carlo + importance sampling run, dividing the former by the estimate of the latter.

\subsection{Stratified sampling with feedforward neural networks}
In addition to importance sampling, stratified sampling is a widely used variance reduction method.
Stratified sampling involves constraining the fraction of trajectories sampled from specific subsets of the sample space.
To implement this method effectively, suitable subsets of the sample space need to be chosen, covering the entire sample space, along with the desired fractions of the overall sample falling within each subset.
It is important to note that stratification typically generates dependent sequences of random variables, which affects the calculation of the standard error and variance of the Monte Carlo estimator.
For further information on this approach, we refer to \cite{MR1999614}.

In \cite{MR1849001}, the authors investigate optimal importance sampling and stratification techniques for pricing path-dependent options.
Similar to \cite{MR2362149}, they employ large deviations techniques to determine asymptotically optimal drift adjustments in a discrete-time framework.
In order to overcome the computational effort that might be required to perform optimal stratification, the authors propose utilizing the drift identified for importance sampling to perform further stratification.
In the following examples, we will augment our results based on importance sampling with the stratified sampling approach.

More precisely, let us consider the estimation of $\mathbb{E}[F(X)]$.
After having discretized the time interval into m points, assume that $F(X)$ can be expressed as a function of $Z$, with $Z$ being a $m$-dimensional vector of independent standard normal variables.
If $f$ denotes the optimal element from Cameron--Martin space, sampled at $t_i$ as a vector, and appropriately re-scaled such that adjusting the drift of $M$ corresponds to adding $f$ to $Z$ in the discrete time case, then we want to sample $Z$ conditional on $f^{\top}Z \in A_{i}$, where $A_i$ denotes a stratum, which is a subset of the sample space.
In our case, $A_i$ is chosen to correspond to the interval between the $(i-1)/N-th$ and the $i/N-th$ quantile of the standard normal distribution, where $N$ denotes the number of strata.
We maintain an equal number of replications for each stratum.
For further details on simulating $Z$ conditional on $f^{\top}Z \in A_{i}$, we refer readers to Section 4 in \cite{MR1849001}.

In Subsections~\ref{subsec:business}-\ref{subsec:dyncor}, we extend our analysis beyond importance sampling by aditionally using the trained neural networks to implement stratified sampling.
By combining these two techniques, we demonstrate the significant potential for further variance reduction.
It is crucial to emphasize that using the optimal importance sampling drift for stratification may not always result in optimal stratification in general.
Furthermore, it is worth noting that the setting of \cite{MR1849001} is discrete in time.
There is ample scope to explore optimal stratified sampling in continuous time using neural networks.

\subsection{Changing business activity}\label{subsec:business}
Methods typically employed for importance sampling based on continuous stochastic processes for asset prices often assume that the dynamics of the asset price are governed by an SDE driven by a Brownian motion.
Here, we aim to deviate from the conventional framework where $\mu = \lambda$, and explore an example involving a time-changed Brownian motion.
It is important to note that in this case, the time-change directly affects the definition of the Cameron--Martin norm through the Lebesgue--Stieltjes measure $\mu$.
The utilization of a deterministic time-change can be interpreted as a means of modeling periods characterized by varying business activity, thus incorporating effects such as seasonality.
See \cite{MR3516141} for an example where this has been done.

Let us consider the asset $X$ governed by the dynamics $\mathrm{d}X_t = r X_t\, \mathrm{d}[M]_t + \sigma X_t\, \mathrm{d}M_t$, where $X_0 = x$ and $M = B_{C_t}$ for $B$ representing a standard Brownian motion.
Motivated by \cite{MR3516141}, we make the assumption that $C_t = \int_{0}^{t} \nu(s)\, \mathrm{d}s$, where the activity rate function $\nu$ takes the form
\begin{align}
    \nu(s) = 
    \begin{cases}
        1 + \kappa(s-0.2) / 0.1, & s \in [0.2, 0.3), \\
        1 + \kappa(0.4-s) / 0.1, & s \in [0.3, 0.4), \\
        1 + 2\kappa(s-0.6) / 0.1, & s \in [0.6, 0.7) \\
        1 + 2\kappa(0.8-s) / 0.1, & s \in [0.7, 0.8) \\
        1, & \mathrm{else},
    \end{cases}
\end{align}
where $\kappa$ denotes the level of business activity.
We further normalize $\nu$ such that $C_{1} = \int_{0}^{1}\nu(s)\, \mathrm{d}s = 1$.
In this case, $\mu$ is absolutely continuous with respect to $\lambda$ with Radon--Nikod\'{y}m density $\nu$, and $[M] = C$.

Figure~\ref{fig:sample_path} illustrates a representative trajectory of $X$ under the assumption of an activity rate function modeled by $\kappa = 10$.
The trajectory exhibits two distinct phases characterized by heightened volatility, which can be interpreted as periods of increased business activity.
Table~\ref{tab:changing_activity} below presents results obtained for various values of $\kappa$.
Note that the special case of $\kappa = 0$ corresponds to the classical Black--Scholes model.

\begin{figure}[htb]
\begin{center}
\includegraphics{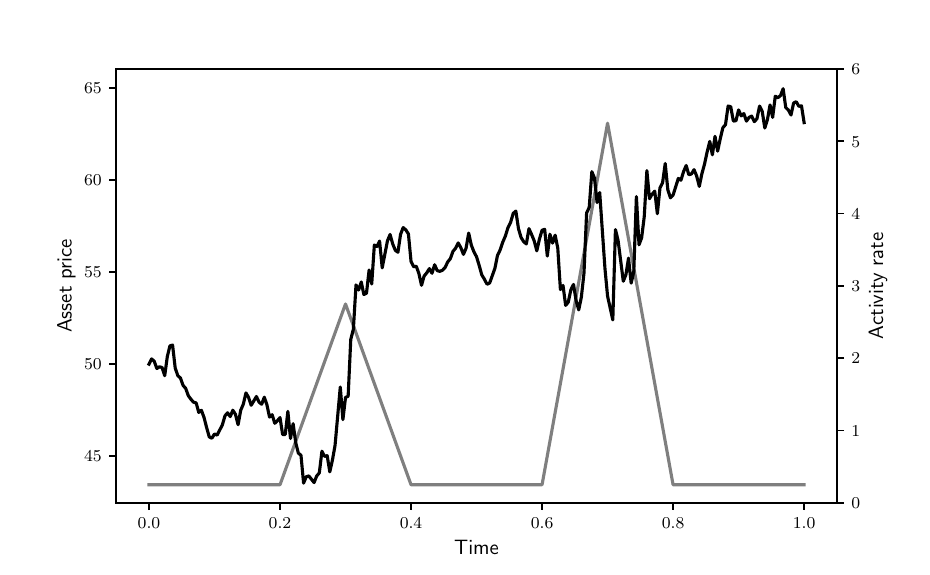}
\vspace*{-0.75cm}
\caption{Typical sample path for the model described above with $\kappa = 10$, along with the corresponding activity rate function $\nu$. Other model parameters are $X_{0} = 50$, $r = 0.05$ and $\sigma = 0.25$.}
\label{fig:sample_path}
\end{center}
\end{figure}

\begin{table}[htb]
    \centering
\begin{threeparttable}
    \caption{Variance ratios for different levels $\kappa$ of business activity.}
    \begin{tabular}{@{} ccccccc @{}}
        \toprule
        \textbf{Parameter} & \multicolumn{3}{c}{\textbf{Importance Sampling}} & \multicolumn{3}{c}{\textbf{IS and Stratification}} \\ 
        $\kappa$ & Mean & Std. err. & Var. ratio & Mean & Std. err. & Var. ratio\\
        \midrule
        0 & 5.945 & 0.019 (0.32 $\%$) & 129 & 5.9337 & 0.0018 (0.03 $\%$) & 14,525 \\
        1 & 4.675 & 0.015 (0.32 $\%$) & 154 & 4.6672 & 0.0023 (0.05 $\%$) & 6,714 \\
        2 & 3.987 & 0.013 (0.33 $\%$) & 167 & 3.9860 & 0.0026 (0.07 $\%$) & 4,191 \\
        5 & 3.053 & 0.010 (0.33 $\%$) & 207 & 3.0495 & 0.0016 (0.05 $\%$) & 8,139\\
        10 & 2.5286 & 0.0086 (0.34 $\%$) & 235 & 2.5304 & 0.0016 (0.06 $\%$) & 7,110 \\ 
        \bottomrule
    \end{tabular}
    \label{tab:changing_activity}
    \begin{tablenotes}\footnotesize
    \item Note: Option prices and standard errors are quoted in cents. Only significant digits are reported. Number of hidden nodes is 2. Other model parameters are $X_{0} = 50$, $r = 0.05$, $\sigma = 0.25$ and $K=70$.
    \end{tablenotes}
\end{threeparttable}
\end{table}

From Table~\ref{tab:changing_activity} it is evident that both importance sampling and the combined approach of importance- and stratified sampling exhibit substantial variance reduction across all values of $\kappa$.
Notably, the combination of importance- and stratified sampling demonstrates a remarkable enhancement in variance reduction compared to using importance sampling alone.

In \cite{MR2362149}, the authors study asymptotically optimal importance sampling in continuous time following a large deviations approach.
In Table 2 of their work, the authors present variance ratios for an arithmetic Asian call option within a Black--Scholes model across various values of volatility (sigma) and strike (K).
We refer to \cite[Section 5]{MR2362149} for details about the model and the selected parameters.
We replicated their Table 2 using neural networks to induce optimal measure changes and subsequently compared the obtained variance ratios.
On average, employing neural networks resulted in a 20\% increase in the variance ratio.
For instance, when considering a volatility of 30\% and a strike of 70, Guasoni \& Robertson report a variance ratio of 56, while our method yielded a variance ratio of 67.

In \cite{MR2451625}, the author studies importance sampling based on a least-squares optimization procedure.
The author presents variance ratios for various combinations of volatility $\sigma$ and strike $K$ in Table 6, specifically for an arithmetic Asian call option within a Black--Scholes model.
Additionally, the table includes variance ratios obtained using an adaptive Robbins--Monro procedure as proposed in \cite{A2004} for the same set of model parameters.
We replicated Table 6 in \cite{MR2451625} using our method and compared the resulting variance ratios.
As it turns out, our method yields average variance ratios that are 10\% and 95\% higher than the values reported by \cite{MR2451625} and \cite{A2004}, respectively.

Finally, in Table 7, Capriotti provides the results for a partial average Asian call option, as previously investigated in \cite{899767}.
For detailed definitions of the models and parameters utilized in the simulations, we refer to Section 5 in \cite{MR2451625}.
We implemented this particular model using our method.
On average, our approach yielded variance ratios that were 10\% and 50\% higher than the values reported by \cite{MR2451625} and \cite{899767}, respectively.

\subsection{Multiple rare events}\label{subsec:knockout}
In \cite{MR1459268}, the authors emphasize that rare events often consist of unions of meaningful events that represent different ways in which the rare event can occur.
In this context, we aim to examine an example where the rare event is formed by the intersection of two rare events.
We will also discuss the case of the union of rare events later on.
An illustrative example is provided by knock-out call options, which exhibit a classical scenario where the payoff is discontinuous with respect to the asset price trajectory.
In this case, two potentially rare events can arise: (1) the arithmetic average $\bar{X}_{t} = \int_{0}^{t}\langle w, X_{s}\rangle\, \mathrm{d}s$ must be above the strike at the terminal time, and (2) the option must not be knocked out.

Consider an asset price $X$ that follows a classical Black--Scholes model, characterized by the SDE $\mathrm{d}X_t = r X_t\, \mathrm{d}t + \sigma X_t\, \mathrm{d}B_t$, with an initial value of $X_0 = x$, where $B$ denotes a Brownian motion.
In contrast to the previous subsection, we introduce knock-out barriers $L, U$ that satisfy $0 < L < X_0 < K < U$.
The option is considered knocked out if the arithmetic average $\bar{X}_{t}$ breaches either of the two barriers at any given point in time before or at maturity.
In our example, there is a delicate balance which needs to be achieved between giving the asset a positive drift such that $\bar{X}_{1} > K$ with sufficiently high probability, and making sure that the option is not knocked out.

In Figure~\ref{fig:learning_evolution} we provide a graphical representation of the learning process of the neural network.
On a fixed dataset, we calculate the probability of the arithmetic average ending up above the strike K, the probability of it remaining between the knock-out barriers at all times, as well as the variance ratio after each epoch that the neural network was trained.
Table~\ref{tab:multiple_events} provides a comprehensive overview of the variance ratios corresponding to different strike prices K and upper knock-out barriers U.

\begin{figure}[htb]
\begin{center}
\includegraphics{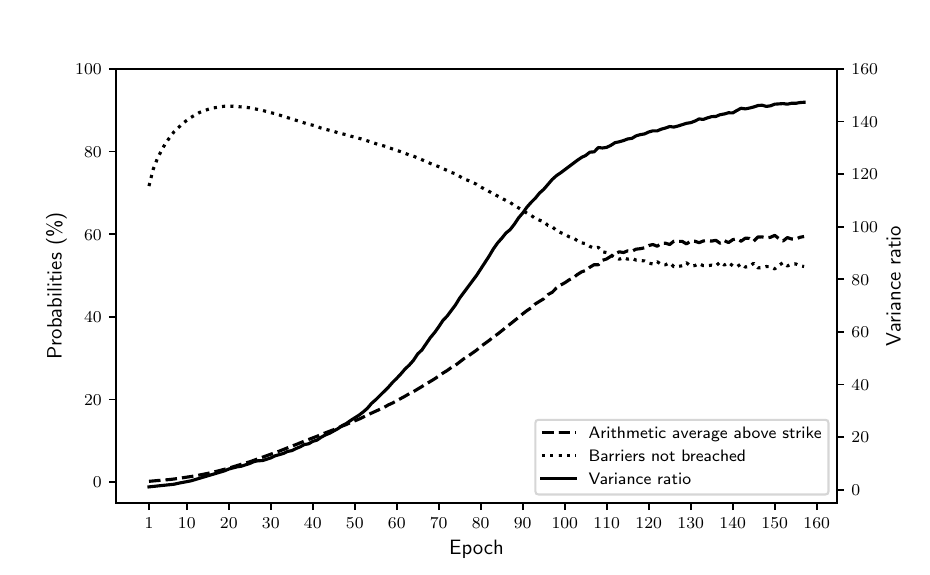}
\vspace*{-0.75cm}
\caption{A graphical representation of the learning process of the neural network.}
\label{fig:learning_evolution}
\end{center}
\end{figure}

\begin{table}[htb]
    \centering
\begin{threeparttable}
    \caption{Variance ratios for different values of strike $K$, volatility $\sigma$ and upper knock-out barrier $U$.}
    \begin{tabular}{@{} ccccccccc @{}}
        \toprule
        \multicolumn{3}{c}{\textbf{Parameters}} & \multicolumn{3}{c}{\textbf{Importance Sampling}} & \multicolumn{3}{c}{\textbf{IS and Stratification}} \\ 
        $K$ & $\sigma$ & $U$ & Mean & Std. err. & Var. ratio & Mean & Std. err. & Var. ratio\\
        \midrule
        \multirow{6}{*}{60} & \multirow{3}{*}{0.2} & 70 & 0.763 & 0.013 (1.70 $\%$) & 7 & 0.779 & 0.012 (1.54 $\%$) & 7 \\
         & & 80 & 12.605 & 0.067 (0.53 $\%$) & 10 & 12.588 & 0.049 (0.39 $\%$) & 18 \\
         & & 90 & 22.607 & 0.078 (0.35 $\%$) & 18 & 22.673 & 0.032 (0.14 $\%$) & 112 \\
         & \multirow{3}{*}{0.3} & 70 & 0.1826 & 0.0082 (4.49 $\%$) & 3 & 0.1816 & 0.0080 (4.41 $\%$) & 3 \\
         & & 80 & 13.65 & 0.12 (0.88 $\%$) & 4 & 13.60 & 0.11 (0.81 $\%$) & 4 \\
         & & 90 & 42.86 & 0.22 (0.51 $\%$) & 5 & 42.89 & 0.16 (0.37 $\%$) & 9\\
        \multirow{6}{*}{70} & \multirow{3}{*}{0.2} & 80 & 0.000775 & 0.000041 (5.29 $\%$) & 356 & 0.000760 & 0.000041 (5.39 $\%$) & 357 \\
         & & 90 & 0.1917 & 0.0018 (0.94 $\%$) & 144 & 0.1921 & 0.0016 (0.83 $\%$) & 189 \\
         & & 100 & 0.6473 & 0.0035 (0.54 $\%$) & 203 & 0.6449 & 0.0021 (0.33 $\%$) & 537 \\
         & \multirow{3}{*}{0.3} & 80 & 0.00070 & 0.00014 (20 $\%$) & 36 & 0.00068 & 0.00014 (20.59 $\%$) & 37 \\
         & & 90 & 0.724 & 0.011 (1.52 $\%$) & 17 & 0.733 & 0.010 (1.36 $\%$) & 18 \\
         & & 100 & 4.513 & 0.034 (0.75 $\%$) & 18 & 4.507 & 0.029 (0.64 $\%$) & 25 \\
        \bottomrule
    \end{tabular}
    \label{tab:multiple_events}
    \begin{tablenotes}\footnotesize
    \item Note: Option prices and standard errors are quoted in cents. Only significant digits are reported. Number of hidden nodes is 2. Other model parameters are $X_{0}=50$, $r = 0.05$ and $L=40$.
    \end{tablenotes}
\end{threeparttable}
\end{table}

Figure~\ref{fig:learning_evolution} highlights an interesting observation: increasing the variance ratio does not simply result from an indiscriminate rise in the probabilities of both rare events occurring.
Instead, it becomes evident that a delicate balance between the occurrence of both rare events is crucial to increase the variance ratio.
As demonstrated in Figure~\ref{fig:learning_evolution}, neural networks exhibit the capability to learn and navigate this balancing act.
Table~\ref{tab:multiple_events} shows again that the neural network-induced change of measure is able to reduce the variance to varying degrees.
We note that compared to the example of the previous subsection, adding stratification does not yield such a dramatic increase in variance ratio, however the improvement is still notable in most cases.

The model which we studied in this subsection has also been explored in \cite{MR1849001}.
In their paper, the authors report in Table 5.2 variance ratios for different values of the volatility, the strike as well as the knock-out barrier $U$ (setting the lower knock-out barrier $L$ to zero).
In contrast to our model, the knock-out occurs in case the asset price breaches the knock-out barrier $U$ at terminal time, i.e. in case $X_{1} > U$.
We replicated their model using our methodology and compared the achieved variance ratios.
Our method on average achieved 20\% higher variance ratios for the case of importance sampling without stratification.
However, when incorporating stratified sampling, our method on average achieved variance ratios that were 10\% lower compared to those reported in \cite[Table 5.2]{MR1849001}.
Note that the setting of \cite{MR1849001} is discrete in time, and that the authors consider asymptotically optimal drift adjustments.
These findings suggest that there might be ample scope to further investigate optimal neural-network induced stratification for continuous-time models.

Let us now revisit the method proposed by \cite{MR2451625}.
In Section 5 of his work, the author presents an example in the form of a European straddle: $F(X) = (X_1 - K)^+ + (K-X_1)^+$.
Capriotti argues that in this case, the optimal sampling density would need to be bi-modal, a property that cannot be effectively captured by a normal distribution.
As we attempted to implement this example, it became evident that the neural network struggled to determine the appropriate drift direction.
This particular instance highlights the challenges associated with relying solely on drift adjustments for variance reduction.
It serves as an example where the rare event can be characterized as the union of two events, shedding light on the limitations of such an approach.

\subsection{Dynamic correlation}\label{subsec:dyncor}
The generality of our paper builds on the decomposition $[M] = \int \pi(s)\mu(\mathrm{d}s)$.
While Subsection~\ref{subsec:business} deviates from the conventional Brownian setting where $\mu = \lambda$, we also aim to present an example that diverges from the typical scenario examined in the existing literature, where $\pi \equiv \mathrm{id}$, representing the identity matrix.
To this end, we consider a Heston model with a dynamic variance-covariance matrix.

We assume that the price process $X$ follows the dynamics given by $\mathrm{d}X_{t} = r X_t \mathrm{d}t + \sqrt{V_t}X_t\, \mathrm{d}B_t$.
The instantaneous variance $V$ follows CIR-type dynamics described by $\mathrm{d}V_t = \kappa(\theta - V_t)\, \mathrm{d}t + \xi\sqrt{V_t}\,\mathrm{d}W_t$.
$B$ and $W$ are correlated Brownian motions related through $\mathrm{d}[B, W]_t = \rho(t)\, \mathrm{d}t$, where the correlation function takes the form $\rho(x) = \bar{\rho} + \bar{\rho}A\sin(2\pi fx)$.
In other words, we deviate from the constant correlations regime by means of the multiple of a sine wave with amplitude A and frequency f.
We present the results for various combinations of amplitude and frequency choices in Table~\ref{tab:dynamic_correlation} below.

\begin{table}[htb]
    \centering
\begin{threeparttable}
    \caption{Variance ratios for different values of amplitude $A$ frequency $f$.}
    \begin{tabular}{@{} cccccccc @{}}
        \toprule
        \multicolumn{2}{c}{\textbf{Parameters}} & \multicolumn{3}{c}{\textbf{Importance Sampling}} & \multicolumn{3}{c}{\textbf{IS and Stratification}} \\ 
        $A$ & $f$ & Mean & Std. err. & Var. ratio & Mean & Std. err. & Var. ratio\\
        \midrule
        0 & 0 & 2.2145 & 0.0085 (0.38 $\%$) & 171 & 2.2308 & 0.0063 (0.28 $\%$) & 311\\
        \multirow{3}{*}{0.2} & 1 & 1.9378 & 0.0076 (0.39 $\%$) & 178 & 1.9517 & 0.0058 (0.30 $\%$) & 304 \\
         & 2 & 2.0807 & 0.0082 (0.39 $\%$) & 171 & 2.0938 & 0.0062 (0.30 $\%$) & 297 \\
         & 4 & 2.1498 & 0.0085 (0.40 $\%$) & 165 & 2.1651 & 0.0065 (0.30 $\%$) & 283 \\
        \multirow{3}{*}{0.5} & 1 & 1.5544 & 0.0062 (0.40 $\%$) & 203 & 1.5666 & 0.0048 (0.31 $\%$) & 338 \\
         & 2 & 1.8926 & 0.0077 (0.41 $\%$) & 173 & 1.9037 & 0.0059 (0.31 $\%$) & 289\\
         & 4 & 2.0553 & 0.0081 (0.38 $\%$) & 168 & 2.0691 & 0.0062 (0.30 $\%$) & 289\\
        \multirow{3}{*}{1} & 1 & 1.0147 & 0.0041 (0.39 $\%$) & 268 & 1.0239 & 0.0032 (0.31 $\%$) & 443 \\
         & 2 & 1.6117 & 0.0064 (0.40 $\%$) & 202 & 1.6230 & 0.0047 (0.29 $\%$) & 369\\
         & 4 & 1.9048 & 0.0078 (0.41 $\%$) & 165 & 1.9160 & 0.0060 (0.31 $\%$) & 277\\
        \bottomrule
    \end{tabular}
    \label{tab:dynamic_correlation}
    \begin{tablenotes}\footnotesize
    \item Note: Option prices and standard errors are quoted in cents. Only significant digits are reported. Number of hidden nodes is 5. Other model parameters are $X_{0}=50$, $r=0.05$, $V_{0}=0.04$, $\kappa=2$, $\theta=0.09$, $\xi=0.2$, $\bar{\rho}=-0.5$ and $K=70$.
    \end{tablenotes}
\end{threeparttable}
\end{table}

\subsection{Basket option}\label{subsec:basket}

So far, we have presented results in scenarios with low dimensions.
However, the multi-dimensional formulation of our setting suggests investigating whether we can achieve satisfactory levels of variance reduction for higher-dimensional models.
Inspired by \cite{MR2569805}, we study a multi-dimensional Black--Scholes model.

Consider the $d$-dimensional asset price $X$ governed by the SDE $\mathrm{d}X_t = r \odot X_t\, \mathrm{d}t + X_t \odot \mathrm{d}M_t$, where $M_t = \Sigma B_t$ represents a $d$-dimensional standard Brownian motion $B_t$ with variance-covariance matrix $\Sigma\Sigma^{\top}$.
We sample the initial value $X_{0}$ of $X$ uniformly from the range of 10 to 200.
Moreover, we sample the vector $r$ of appreciation rates and the vector $\sigma$ of volatilities uniformly between 1\% and 9\% as well as 10\% and 30\%, respectively.
The weight vector $w$ is then computed as $w_i = r_i / \sigma_i^2$ and further normalized to sum to 1.

To define the matrix $\Sigma\Sigma^{\top}$, it is necessary to specify the correlation matrix.
In order to ensure a valid correlation matrix that remains positive definite even in high dimensions, we adopt the approach proposed in \cite{MR1799307}.
Firstly, we sample a d-dimensional vector $y$ uniformly between $0$ and $1$.
We then re-scale the vector $y$ such that the sum of its elements equals the dimension $d$.
The algorithm proposed in \cite{MR1799307} then generates a valid correlation matrix, whose eigenvalues correspond to the values in the re-scaled vector $y$.
Finally, we still need to specify the strike.
To this end, we sample $10^4$ observations of the arithmetic average $\bar{X}$ at maturity, and then choose the strike K to approximately be above the 90th percentile of the distribution of $\bar{X}_{1}$.
Note that the choice of K is highly dependent on the previously sampled parameters.

\begin{table}[htb]
    \centering
\begin{threeparttable}
    \caption{Variance ratios for different dimensions $d$.}
    \begin{tabular}{@{} ccccccc @{}}
        \toprule
        \multicolumn{2}{c}{\textbf{Parameters}} &  \multicolumn{5}{c}{\textbf{Importance Sampling}} \\ 
        $d$ & $K$ & Mean & Std. err. & Var. ratio & $\mathbb{P}(F(X) > 0)$ & $\mathbb{Q}(F(X) > 0)$\\
        \midrule
        10 & 88 & 3.7557 & 0.0118 (0.31 $\%$) & 62 & 3.24 $\%$ & 70.83 $\%$ \\
        20 & 115 & 1.3252 & 0.0049 (0.37 $\%$) & 124 & 1.22 $\%$ & 65.42 $\%$ \\
        50 & 126 & 6.4457 & 0.0189 (2.93 $\%$) & 28 & 6.96 $\%$ & 72.36 $\%$ \\
        100 & 106 & 1.6995 & 0.0056 (0.33 $\%$) & 54 & 3.36 $\%$ & 70.69 $\%$ \\
        200 & 112 & 2.1373 & 0.0067 (0.31 $\%$) & 36 & 5.17 $\%$ & 72.04 $\%$ \\
        500 & 110 & 1.1352 & 0.0042 (0.37 $\%$) & 30 & 4.46 $\%$ & 74.99 $\%$ \\
        1,000 & 110 & 2.327 & 0.011 (0.47 $\%$) & 6 & 10.87 $\%$ & 84.72 $\%$ \\ 
        \bottomrule
    \end{tabular}
    \label{tab:basket_option}
    \begin{tablenotes}\footnotesize
    \item Note: Option prices and standard errors are quoted in cents. Only significant digits are reported. Number of hidden nodes corresponds to the dimension $d$.
    $\mathbb{P}(F(X) > 0)$ represents the proportion of trajectories in the test dataset where the payoff is positive, without incorporating a drift adjustment. $\mathbb{Q}(F(X) > 0)$ denotes the proportion of trajectories in the test dataset where the payoff is positive under the drift adjustment.
    \end{tablenotes}
\end{threeparttable}
\end{table}

Table~\ref{tab:basket_option} presents variance ratios obtained for various dimensions $d$ ranging from $d=10$ up to $d=1{,}000$.
Moreover, we also compared our method to the approach presented in \cite{MR2569805}.
In their study, the authors considered the 40-dimensional case, and all volatilities, appreciation rates, and weights were chosen uniformly across all assets in the basket. 
We refer to Section 3 in \cite{MR2569805} for further details about the model as well as model parameters in their Table 1.
As it turns out, our method achieves variance ratios that are, on average, comparable to those reported by \cite{MR2569805}.
It is important to note that the strikes which were chosen are relatively close to the initial value.
In previous examples, we can observe that the obtained variance ratios tend to grow as the strike in increased.
In contrast to \cite{MR2569805}, we present in Table \ref{tab:basket_option} results for dimension up to $d=1{,}000$, which we believe is a distinctive aspect worth highlighting.

\section{Conclusion}\label{sec:conclusion}

In this paper, we presented a method that uses feedforward neural networks for the purpose of reducing the variance of Monte Carlo estimators.
To this end, we studied the class of Gaussian measures which are induced by vector-valued continuous local martingales with deterministic covariation.
Building on the theory of vector stochastic calculus, we identified the Cameron--Martin spaces of those measures, and proved universal approximation theorems that establish, up to an isometry, topological density of feedforward neural networks in these spaces.
We then applied our results to a classical importance sampling approach which seeks for an optimal drift adjustment of the processes which are driving the asset prices.
Finally, we presented the results of a numerical study, which clearly indicate the potential of this approach.

Let us also remark that our approach comes with several challenges.
In principle, one needs to train separate feedforward networks for different models and model parameters.
In light of Remark~\ref{rem:sen}, one could train one feedforward network to minimize a weighted average standard error over several models or model parameters.
Complex, high-dimensional models might call for the need of using complex neural network architectures in order to achieve a sufficient variance reduction, which might lead to a considerable computational effort for training the feedforward networks.
On the other hand, the competing approaches \cite{MR2362149,MR2565852} involve having to solve a potentially complex, high-dimensional variational problem, whose solution might involve a numerical procedure which might induce a considerable computational effort, too.
Finally, while Theorem~\ref{thm:densityapprox} and the simulations of Section~\ref{sec:simulation} show that one can obtain a sufficient variance reduction with shallow feedforward networks, the model-dependent choice of optimal architecture has not been discussed at all, which highlights the potential for a further improvement of this method.

\subsection{Outlook on further research}\label{subsec:outlook}

Throughout this paper, we assumed for the process $M$ to be a continuous local martingale with deterministic covariation, such that it is a Gaussian process and induces a Gaussian measure on path space.
Clearly, there are Gaussian processes which cannot be local martingales, e.g. fractional Brownian motion with Hurst index $\neq 1/2$.
In line with Remark~\ref{rem:fbm}, Section~\ref{sec:CMspace} can be extended to the study of multivariate Volterra type Gaussian processes of the form $\tilde{M}_{t} = \int_{0}^{u} k(t,s)\, \mathrm{d}M_{s}$ with a matrix-valued kernel $k$.
While Section~\ref{sec:IS} makes use of the semimartingale property of $M$ by applying Girsanov's theorem and studying convergence of stochastic exponentials, the Cameron--Martin theorem (see Theorem~\ref{app:CMthm}) can still be applied to the Gaussian measure that is induced by $\tilde{M}$ on path space.
These considerations in particular motivate the study of a refined class of multivariate (fractional) stochastic volatility models, their small-time asymptotics as well as importance sampling methods for the numerical evaluation of derivatives for these models, which is subject to a follow-up work.

In Section~\ref{sec:IS}, we required the random payoff $F(X)$ to be $\Fcal_{u}$-measurable and $L^{p}$-integrable for some $p > 2$.
However, the properties that we imposed on the process $X$ were rather weak.
In particular, Theorem~\ref{thm:densityapprox} only considered $F(X)$ as a random variable, where we used the SDE for $X$ only when performing a measure change and applying Girsanov's theorem in order to understand the semimartingale decomposition of $X$ under a new sampling measure.
Therefore, the methods from Section~\ref{sec:IS} should extend to the case where $X$ is the solution to a McKean--Vlasov SDE, provided that we understand how the dynamics of the process change under a change of measure.
We leave it to a follow-up work to combine our methods with ideas from \cite{dRST2018}, which should lead to a tractable importance sampling framework for the evaluation of derivatives on solutions to McKean--Vlasov SDEs under weaker assumptions then those imposed on \cite{dRST2018}.

The setting of this paper does naturally apply to the evaluation of European options and asset price processes with continuous paths. 
More generally, reducing the standard error of Monte Carlo estimators with neural networks when pricing American options based on the popular algorithm proposed by Longstaff \& Schwartz (cf. \cite{MR1932380,LS2021}) and models with jumps, very much in the spirit of \cite{MR4035022, MR2475940}, provides another interesting challenge that is reserved for follow-up work.

Finally, the measure changes which we studied in Section~\ref{sec:IS} were induced by density processes of the form $\mathcal{E}(f^{\top}\sbullet[.75]M)$, where $f \in \Lambda^{2}$ is a deterministic function.
The reason why we did not consider the more general class of processes $U \in L^{2}(M)$ for which $\mathcal{E}(U^{\top}\sbullet[.75]M)$ is a martingale is twofold.
While the proof of Theorem~\ref{thm:densityapprox} would become more involved, one would need to use neural network architectures which are more complex then the ones which we discussed in Section~\ref{sec:ffn}.
For this reason, we argue that the problem of considering deterministic functions $f \in \Lambda^{2}$ provides a tractable, numerically efficient method to reduce the variance in Monte Carlo simulations, and reserve the extension to processes $U \in L^{2}(M)$ and their approximation with neural networks for future work.

\section*{Acknowledgements}
Aleksandar Arandjelovi{\'c} acknowledges support from the  International Cotutelle Macquarie University Research Excellence Scholarship.

\appendix

\section{Gaussian measures}\label{ap:GM}

In this part of the appendix, we collect for the readers' convenience some classical definitions and results about Gaussian measures, for which we mostly rely on the excellent monographs \cite{MR1642391,MR3024389,MR2760872}.
Let $(E, \|\4\cdot\4\|_{E})$ denote a real separable Banach space, $\gamma$ a Borel probability measure on $E$ and $M = (M_{t})_{t \in T}$ an $\re^{d}$-valued process on a probability space\/ $(\Omega, \Fcal, \PP)$.
Given $h \in E$, we further denote by $\gamma_{h}$ the measure on $E$ that is induced by the translation $E \ni x \mapsto x + h$.

\begin{definition}\label{app:GM}
The measure\/ $\gamma$ is called Gaussian, if each\/ $f \in E^{*}$ induces a Gaussian distribution on\/ $(\re, \B_{\re})$.
The process\/ $M = (M_{t})_{t \in T}$ is called Gaussian, if\/ $(M_{t_{i}})_{i=1}^{n}$ is jointly Gaussian for each\/ $n \in \na$ and\/ $0 \le t_{1} < t_{2} < \hdots < t_{n} \le u$.
\end{definition}

A Gaussian measure $\gamma$ is centered, if each $f \in E^{*}$ induces a centered Gaussian distribution.
Similarly, a Gaussian process $M = (M_{t})_{t \in T}$ is centered, if $(M_{t_{i}})_{i=1}^{n}$ is jointly centered Gaussian for each\/ $n \in \na$ and\/ $0 \le t_{1} < t_{2} < \hdots < t_{n} \le u$.
Since Section \ref{sec:CMspace} only considers centered Gaussian processes and measures, we will from now on restrict to this special case.

In the context of Definition \ref{app:GM}, we have the natural embedding $j \from E^{*} \to E_{\gamma}^{*}$, where $E_{\gamma}^{*}$ denotes the reproducing kernel Hilbert space of $\gamma$, which is defined as the closure of $E^{*}$ in $L^{2}(\gamma)$.
We further define the covariance operator of $\gamma$ by the map
\begin{align*}
R_{\gamma} \from E^{*} \to (E^{*})' \colon f \mapsto \Bigl( g \mapsto \int_{E}f(x)g(x)\, \gamma(\mathrm{d}x) \Bigr),
\end{align*}
and implicitly consider its extension to $E_{\gamma}^{*}$, i.e. $R_{\gamma} \from E_{\gamma}^{*} \to (E^{*})'$.

Given $f \in E_{\gamma}^{*}$, note that $R_{\gamma}(f) \from E^{*} \to \re$ is a linear operator.
If we endow $E^{*}$ with the Mackey topology, then \cite[Lemma 3.2.1]{MR1642391} shows that $R_{\gamma}(f)$ is continuous.
Mackey's theorem (cf. \cite[Theorem A 1.1]{MR1642391}) then yields the existence of $x_{f} \in E$, such that $R_{\gamma}(f)(g) = g(x_{f})$ for each $g \in E^{*}$.
We then also denote by $R_{\gamma}$ the map $E_{\gamma}^{*} \ni f \mapsto x_{f}$.  

\begin{definition}[Cameron--Martin space, cf. {\cite[Lemma 2.4.1]{MR1642391}}]\label{app:CMspace}
The Cameron--Martin space\/ $H(\gamma)$ of a centered Gaussian measure\/ $\gamma$ is defined as the range of\/ $R_{\gamma}$ in $E$, i.e. $H(\gamma) \coloneqq R_{\gamma}(E_{\gamma}^{*}) \subset E$.
We equip\/ $H(\gamma)$ with the inner product
\begin{align*}
\langle h, k \rangle_{H(\gamma)} \coloneqq \langle \hat{h}, \hat{k} \rangle_{L^{2}(\gamma)} = \int_{E} \hat{h}(x)\hat{k}(x)\, \gamma(\mathrm{d}x),
\quad h,k \in H(\gamma),
\end{align*}
where $h = R_{\gamma}(\hat{h})$\/ and\/ $k = R_{\gamma}(\hat{k})$ for some\/ $\hat{h}, \hat{k} \in E_{\gamma}^{*}$.
\end{definition}

The space $(H(\gamma), \langle \cdot, \cdot \rangle_{H(\gamma)})$ is a real separable Hilbert space that is continuously embedded into $E$ (cf. \cite[Proposition 2.4.6 and Theorem 3.2.7]{MR1642391}).
Moreover, \cite[Theorem 2.4.7]{MR1642391} shows that $H(\gamma)$ is of $\gamma$-measure zero, whenever $E_{\gamma}^{*}$ is infinite dimensional.

\begin{remark}[Degenerate case]\label{app:degenerate}
Given a centered Gaussian measure $\gamma$ on $E$, the topological support of $\gamma$ is defined as the smallest closed subset $S \subset E$ with $\gamma(E \setminus S) = 0$, and is given by $\overbar{H(\gamma)}$, where the closure is taken in $E$ (cf. \cite[Theorem 3.6.1]{MR1642391}).
We call $\gamma$ nondegenerate, if $\overbar{H(\gamma)} = E$, or equivalently, if $H(\gamma)$ is densely embedded into $E$.
If $\overbar{H(\gamma)}$ is a strict subspace of $E$, then we call $\gamma$ degenerate.
\end{remark}

\begin{remark}
If $\gamma$ and $\tilde{\gamma}$ are two centered Gaussian measures on $E$, such that $H(\gamma) = H(\tilde{\gamma})$ and $\|\4\cdot\4\|_{H(\gamma)} = \|\4\cdot\4\|_{H(\tilde{\gamma})}$, then $\gamma$ and $\tilde{\gamma}$ coincide (cf. \cite[Corollary 3.2.6]{MR1642391}).
Moreover, if $E$ is continuously and linearly embedded into another real separable Banach space $\tilde{E}$ with embedding $i$ and induced Gaussian measure $\nu = \gamma \circ i^{-1}$, then $\tilde{E} \supset H(\nu) = i(H(\gamma))$ (cf. \cite[Lemma 3.2.2]{MR1642391}).
\end{remark}

The Cameron--Martin space has another useful characterization, which is stated in the following theorem.

\begin{theorem}[Cameron--Martin, cf. {\cite[Theorem 2.4.5]{MR1642391} and \cite[Theorem 1]{MR10346}}]\label{app:CMthm}
Given a centered Gaussian measure\/ $\gamma$ on $E$ and\/ $h \in E$, the measures\/ $\gamma$ and\/ $\gamma_{h}$ are equivalent precisely when\/ $h \in H(\gamma)$, and singular otherwise.
In particular,
\begin{align*}
H(\gamma) = R_{\gamma}(E_{\gamma}^{*}) = \bigl\{h \in E\, |\, \gamma_{h} \sim \gamma\bigr\}.
\end{align*}
\end{theorem}

Whenever $\gamma$ is Gaussian, the measure $\gamma_{h}$ is Gaussian for each $h \in E$ (cf. \cite[Lemma 2.2.2]{MR1642391}).
Consequently, Theorem \ref{app:CMthm} characterizes a set of Gaussian measures which are equivalent to $\gamma$.
The following theorem is another central result, which in particular implies that $\gamma_{h}$ and $\gamma$ are singular whenever $h \in E \setminus H(\gamma)$.

\begin{theorem}[Feldman--Hajek, cf. {\cite[Theorem 2.7.2]{MR1642391}}]
Any two Gaussian measures on\/ $E$ are either equivalent or mutually singular.
\end{theorem}

In order to quantify the (exponential) decline of the probability of certain tail events, the following result is often times useful.

\begin{proposition}[Large deviation principle, cf. {\cite[Corollary 4.9.3]{MR1642391}}]\label{app:LDP}
Let\/ $\gamma$ denote a centered Gaussian measure on\/ $E$.
Moreover, for\/ $\varepsilon > 0$, let\/ $\gamma_{\varepsilon}$ denote the pushforward measure of\/ $\gamma$ under the map\/ $E \ni f \mapsto \sqrt{\varepsilon}f$.
Then,\/ $(\gamma_{\varepsilon})_{\varepsilon > 0}$ satisfies the large deviation principle with rate function\/ $I_{\gamma} \from E \to \Rbarplus$, where
\begin{align*}
    I_{\gamma}(f)=
    \begin{cases}
        \tfrac{1}{2}\|f\|_{H(\gamma)}^{2} &\mathrm{for}\, f \in H(\gamma),\\
        \infty & \mathrm{otherwise}.
    \end{cases}
\end{align*}
In other words, for each $F \in \B_{E}$,
\begin{align*}
    - \inf_{f \in F^{\circ}} I_{\gamma}(f) \le \liminf_{\varepsilon \searrow 0}\varepsilon\log\gamma_{\varepsilon}(F) \le \limsup_{\varepsilon \searrow 0}\varepsilon\log\gamma_{\varepsilon}(F) \le - \inf_{f \in \overbar{F}}I_{\gamma}(f).
\end{align*}
\end{proposition}

Before we finish this section, we state a result that allows us in many cases to obtain a tractable representation of $(H(\gamma), \langle \cdot, \cdot \rangle_{H(\gamma)})$.

\begin{theorem}[Factorization, cf. {\cite[Section 3.3]{MR1642391} and \cite[Section~4.2]{MR3024389}}]\label{thm:factorization}
Given a centered Gaussian measure\/ $\gamma$ on $E$, assume that there exists a Hilbert space\/ $\tilde{H}$ and a continuous linear operator\/ $J \from \tilde{H} \to E$ such that\/ $R_{\gamma}$ admits the factorization\/ $R_{\gamma} = J \circ J^{*}$, where\/ $J^{*} \from E^{*} \to \tilde{H}^{*} \cong \tilde{H}$ denotes the adjoint of\/ $J$.
Then\/ $H(\gamma)$ coincides with\/ $J(\tilde{H})$.
If\/ $J$ is moreover injective, then
\begin{align*}
\langle f, g \rangle_{H(\gamma)} = \langle J^{-1}(f), J^{-1}(g) \rangle_{\tilde{H}},
\quad f,g \in H(\gamma).
\end{align*}
\end{theorem}

\section{Proofs}\label{app:technical}

\begin{proof}[Proof of Lemma \ref{lem:L2space1}]
The proof of \cite[Lemma 3.2]{MR1975582} reveals that $\|\4\cdot\4\|_{\Lambda^{2}}$ satisfies the triangle inequality, which shows that $\Lambda^{2}$ is a real vector space.
In order to see that $\langle\cdot,\cdot\rangle_{\Lambda^{2}}$ is an inner product on $\Lambda^{2}$, note that, by construction, $\langle\cdot,\cdot\rangle_{\Lambda^{2}}$ is symmetric and linear in both arguments, and recall that $\pi$ is positive semidefinite $\mu$-a.e., hence $f^{\top}\pi f \ge 0$\/ $\mu$-a.e. and therefore $\langle f, f \rangle_{\Lambda^{2}} \ge 0$ for each measurable $f \from T \to \re^{d}$.
If $\langle f, f \rangle_{\Lambda^{2}} = 0$ for some $f \in \Lambda^{2}$, then $f^{\top}\pi f = 0$\/ $\mu$-a.e., hence $f \sim 0$, which implies that $\langle\cdot,\cdot\rangle_{\Lambda^{2}}$ is positive definite and therefore an inner product on $\Lambda^{2}$.

Completeness of $(\Lambda_{2}, \varrho_{2})$, where $\varrho_{2}$ denotes the translation invariant metric induced by $\|\4\cdot\4\|_{\Lambda^{2}}$, follows from \cite[Lemme 4.29]{MR542115}, and separability can be argued by adapting the proofs of \cite[Theorem 19.2]{MR2893652} and \cite[Lemma 3.2]{MR1975582}.
We conclude that $(\Lambda^{2}, \langle\cdot,\cdot\rangle_{\Lambda^{2}})$ is a real separable Hilbert space.
\end{proof}

Lemma \ref{lem:L2space2} is a direct consequence of Fr\'{e}chet--Riesz's representation theorem, since we know by Lemma \ref{lem:L2space1} that $\Lambda^{2}$ is a Hilbert space.

\begin{proof}[Proof of Lemma \ref{lem:L2space3}]
$\Lambda^{2,0}$ is clearly a real vector space.
Given $f \from T \to \re^{d}$ measurable, and $i,j \in \{1,2,\hdots,d\}$, by a version of the Kunita--Watanabe inequality for Lebesgue--Stieltjes integrals (cf. \cite[Lemma 5.89]{Schmock2021}),
\begin{align}\label{eq:KWLS}
\begin{aligned}
\Bigl( \bigl| \int_{T} f_{i}(s)f_{j}(s)\, \mu_{i,j}(\mathrm{d}s) \bigr| \Bigr)^{2} & \le \Bigl( \int_{T} |f_{i}(s)f_{j}(s)|\, |\mu_{i,j}|(\mathrm{d}s) \Bigr)^{2} \\
& \le \int_{T} f_{i}^{2}(s)\, \mu_{i,i}(\mathrm{d}s) \int_{T} f_{j}^{2}(s)\, \mu_{j,j}(\mathrm{d}s) < \infty,
\end{aligned}
\end{align}
hence $\Lambda^{2,0} \subset \Lambda^{2}$, and $(\Lambda^{2,0}, \langle \cdot,\cdot \rangle_{\Lambda^{2}})$ is therefore an inner product space.

By dominated convergence, the bounded measurable functions $f \from T \to \re^{d}$ are dense in $\Lambda^{2}$.
But since these functions are contained in $\Lambda^{2,0}$, we see that $\Lambda^{2,0}$ is dense in $\Lambda^{2}$, which also implies the separability of $\Lambda^{2,0}$.
From Lemma \ref{lem:L2space1} we further know that $(\Lambda^{2}, \langle\cdot,\cdot\rangle_{\Lambda^{2}})$ is a Hilbert space and in particular complete.
This shows \ref{lem:L2space3b}.
\end{proof}

\begin{proof}[Proof of Lemma \ref{lem:L2space4}]
The continuity of the embedding follows from \eqref{eq:KWLS} and \eqref{eq:helper} below.
The remaining assertion follows from a multivariate version of \cite[Lemma 1.37]{MR4226142}.
\end{proof}

\begin{proof}[Proof of Lemma \ref{lem:L2space5}]
Let $(\tilde{\pi}, \tilde{\mu})$ be another pair that satisfies \eqref{eq:covarid}, and $f \from T \to \re^{d}$ be measurable.
We then have $\mathrm{d}\mu_{i,j} / \mathrm{d}\mu = \pi_{i,j}$ as well as $\mathrm{d}\mu_{i,j} / \mathrm{d}\tilde{\mu} = \tilde{\pi}_{i,j}$ for all $i,j \in \{1,2,\hdots,d\}$, hence
\begin{align*}
\int_{T}f^{\top}(s)\pi(s)f(s)\, \mu(\mathrm{d}s) = \sum_{i,j=1}^{d} \int_{T}f_{i}(s)f_{j}(s)\, \mu_{i,j}(\mathrm{d}s) = \int_{T}f^{\top}(s)\tilde{\pi}(s)f(s)\, \tilde{\mu}(\mathrm{d}s),
\end{align*}
and we see that $\|\4\cdot\4\|_{\Lambda^{2}}$ does not depend on the specific choice of $(\pi, \mu)$ satisfying \eqref{eq:covarid}.
Moreover, for $f,g \from T \to \re^{d}$ measurable, $(f-g)^{\top}\pi(f-g) =0$\/ $\mu$-a.e. holds precisely when $\|f-g\|_{\Lambda^{2}} = 0$ which is equivalent to $(f-g)^{\top}\tilde{\pi}(f-g) =0$\/ $\tilde{\mu}$-a.e.
\end{proof}

\begin{proof}[Proof of Proposition \ref{prop:CMspace1}]
By a variant of the Cauchy--Schwarz inequality, for each $h \in H$,\/ $i \in \{1,2,\hdots,d\}$ and $t \in T$, it holds that
\begin{align}\label{eq:CSvariant2}
\int_{[\4 0,t]} \bigl|\sum_{j=1}^{d} \pi_{i,j}(s)f_{h,j}(s)\bigr|\, \mu(\mathrm{d}s) \le \sqrt{\mu_{i,i}([\4 0,t])}\|f_{h}\|_{\Lambda^{2}} \le \sqrt{\mu(T)}\|h\|_{H},
\end{align}
which shows that the integral in \eqref{CMfunc} is well defined.
\end{proof}

\begin{proof}[Proof of Proposition \ref{prop:CMspace2}]
$H$ is clearly a real vector space and $\langle \cdot, \cdot \rangle_{H}$ is symmetric and linear in both arguments.
In order to show that $\langle \cdot, \cdot \rangle_{H}$ is positive definite, note that $\langle h, h \rangle_{H} = \| f_{h} \|_{\Lambda^{2}}^{2} \ge 0$ for each $h \in H$.
If $h \in H$ satisfies $\langle h, h \rangle_{H} = 0$, then $f_{h}^{\top}\pi f_{h} = 0$\/ $\mu$-a.e., hence $f_{h} \sim 0$.
An application of inequality \eqref{eq:CSvariant2} shows that $(\pi f_{h})_{i} = 0$\/ $\mu$-a.e. for each $i \in \{1,2,\hdots, d\}$, hence $h = 0$.
$(H, \langle\cdot,\cdot\rangle_{H})$ is therefore an inner product space.

We obtain a norm $\|\4\cdot\4\|_{H}$ on $H$ by setting $\| h \|_{H} \coloneqq \sqrt{\langle h, h \rangle}_{H}$ and thus also a metric $\varrho_{H}$ on $H$ by setting $\varrho_{H}(f, g) \coloneqq \| f-g \|_{H}$.
In order to see that $(H, \varrho_{H})$ is complete, let $(h_{n})_{n \in \na}$ be a Cauchy sequence in $H$.
Then $(f_{h_{n}})_{n \in \na}$ is a Cauchy sequence in $\Lambda^{2}$.
From Lemma \ref{lem:L2space1} we know that $\Lambda^{2}$ is complete.
Consequently, there exists some $f \in \Lambda^{2}$ such that $f_{h_{n}} \to f$ in $\Lambda^{2}$.
If we set $h = J(f)$, then $h \in H$ and $h_{n} \to h$ in $H$.

Finally, in order to see that $H$ is separable, note first that $\Lambda^{2}$ is separable by Lemma \ref{lem:L2space1}.
But this already shows that $H$ is separable as well, because a countable dense subset of $H$ is given by $\{ h \in H \colon f_{h} \in B \}$, where $B$ is a countable dense subset of $\Lambda^{2}$.
\end{proof}

\begin{proof}[Proof of Proposition \ref{prop:CMspace3}]
By construction, $J \from \Lambda^{2} \to H$ is a linear isometry.
Since $\Lambda^{2,0}$ is a linear subspace of $\Lambda^{2}$ by Lemma \ref{lem:L2space3}, we see that $(H^{0}, \langle\cdot,\cdot\rangle_{H})$ is an inner product subspace of $H$.
If $(h_{n})_{n \in \na}$ is a Cauchy sequence in $H^{0}$, then $(f_{h_{n}})_{n \in \na}$ is a Cauchy sequence in $\Lambda^{2,0}$.
By Lemma \ref{lem:L2space3} there exists an $f \in \Lambda^{2}$ such that $f_{h_{n}} \to f$ as $n \to \infty$.
Denoting $h = J(f) \in H$, it follows that $h_{n} \to h$ in $H$, as $n \to \infty$.
\end{proof}

Proposition \ref{prop:CMspace5} is a direct consequence of Fr\'{e}chet--Riesz's representation theorem, since we know by Proposition \ref{prop:CMspace2} that $H$ is a Hilbert space.

\begin{remark}\label{rem:multidual}
For the proof of Proposition \ref{prop:CMspace6}, we will make use of a multivariate version of Riesz--Markov--Kakutani's representation theorem, which states that every $F \in (C(T;\re^{d}))^{*}$ may be identified with an $\re^{d}$-valued set function $\nu = (\nu_{1}, \nu_{2},\hdots, \nu_{d})^{\top}$ on $\B_{T}$, where every entry is a signed Borel measure of finite total variation, such that
\begin{align*}
F(f) = \sum_{j=1}^{d} \int_{T} f_{j}(s)\, \nu_{j}(\mathrm{d}s) \eqqcolon \int_{T} f^{\top}(s)\, \nu(\mathrm{d}s),
\quad f \in C(T;\re^{d}).
\end{align*}
Given $f \from T \to \re^{n \times d}$ such that $(f_{i, \cdot})^{\top}\in C(T; \re^{d})$ for each $i \in \{1,2,\hdots,n\}$, we write
\begin{align*}
\int_{T}f(s)\, \nu(\mathrm{d}s) = \Bigl( \int_{T} f_{1, \cdot}(s)\, \nu(\mathrm{d}s), \hdots, \int_{T} f_{n, \cdot}(s)\, \nu(\mathrm{d}s) \Bigr)^{\top}.
\end{align*}
For generalizations to infinite-dimensional domains and image spaces, see \cite{G1936,MR96964}.
\end{remark}

\begin{proof}[Proof of Proposition \ref{prop:CMspace6}]
We argue in line with \cite[Example~4.4]{MR3024389} and use Theorem \ref{thm:factorization}.
First, we note that every $h \in H$ is continuous and satisfies $h(0) = 0$, hence $H \subset E$.
Let us consider $J$ as a linear operator onto $E$, i.e. $J \from \Lambda^{2} \to E$, which is continuous due to \eqref{eq:CSvariant2}.
In the context of Remark \ref{rem:multidual}, $E^{*}$ is given as the quotient space, where we identify those $\nu \in (C(T;\re^{d}))^{*}$ that annihilate $E$.

For $f \sim \sigma$ and $g \sim \nu$ in $E^{*}$,
\begin{align*}
R_{\gamma}(f)(g) & = \int_{E}f(x)g(x)\, \gamma_{M}(\mathrm{d}x) = \mathbb{E}[f(M)g(M)] = \mathbb{E}\Bigl[ \int_{T}M_{s}^{\top} \sigma(\mathrm{d}s)\int_{T}M_{s}^{\top} \nu(\mathrm{d}s)\Bigr] \\
& = \sum_{i,j=1}^{d} \int_{T} \int_{T} \mathbb{E}[M_{s}^{i}M_{t}^{j}]\, \sigma_{i}(\mathrm{d}s)\, \nu_{j}(\mathrm{d}t) = \sum_{i,j=1}^{d} \int_{T} \int_{T} [M]_{s \wedge t}^{i,j}\, \sigma_{i}(\mathrm{d}s)\, \nu_{j}(\mathrm{d}t) \\
& = \sum_{j=1}^{d} \int_{T} \sum_{i=1}^{d} \int_{T} [M]_{s \wedge t}^{i,j}\, \sigma_{i}(\mathrm{d}s)\, \nu_{j}(\mathrm{d}t) = \sum_{j=1}^{d} \int_{T} \Bigl( \int_{T}[M]_{s \wedge t}\, \sigma(\mathrm{d}s) \Bigr)_{j} \nu_{j}(\mathrm{d}t) \\
& = \int_{T} \Bigl( \int_{T}[M]_{s \wedge t}\, \sigma(\mathrm{d}s) \Bigr)^{\top} \nu(\mathrm{d}t) = \bigl( g, \int_{T}[M]_{s \wedge \cdot}\, \sigma(\mathrm{d}s) \bigr).
\end{align*}
We can therefore identify $R_{\gamma}(f)$ with $\int_{T}[M]_{s \wedge \cdot}\, \sigma(\mathrm{d}s)$.

Next, let us find the adjoint of $J$.
Given $f \in \Lambda^{2}$ and $g \sim \nu$ in $E^{*}$,
\begin{align*}
\bigl( g, J(f) \bigr) & = \sum_{i=1}^{d} \int_{T} J_{i}(f)(t)\, \nu_{i}(\mathrm{d}t) = \sum_{i=1}^{d} \int_{T} \int_{[\4 0, t]}\pi_{i, \cdot}(s)f(s)\, \mu(\mathrm{d}s)\, \nu_{i}(\mathrm{d}t) \\
& = \sum_{i=1}^{d} \int_{T} \int_{T}\indicatorset{[\4 0,t]}(s) \pi_{i, \cdot}(s)f(s)\, \mu(\mathrm{d}s)\, \nu_{i}(\mathrm{d}t) \\
& = \sum_{i=1}^{d} \int_{T} \int_{T}\indicatorset{[\4 0,t]}(s)\, \nu_{i}(\mathrm{d}t)\pi_{i, \cdot}(s)f(s)\, \mu(\mathrm{d}s) \\
& = \sum_{i=1}^{d} \int_{T} \nu_{i}\bigl([s,u]\bigr)\pi_{i, \cdot}(s)f(s)\, \mu(\mathrm{d}s) \\
& = \int_{T} \nu\bigl([s,u]\bigr)^{\top}\pi(s)f(s)\, \mu(\mathrm{d}s) = \langle \nu([\cdot, u]), f \rangle_{\Lambda^{2}},
\end{align*}
hence the adjoint $J^{*} \from E^{*} \to (\Lambda^{2})^{*} \cong \Lambda^{2}$ is given by $g \sim \nu \mapsto (T \ni s \mapsto \nu\bigl([s,u]\bigr))$.

Finally, the covariance operator admits for $g \sim \nu$ in $E^{*}$ the factorization
\begin{align*}
R_{\gamma}(g)_{i}(t) & = \int_{T}[M]_{s \wedge t}^{i, \cdot}\, \nu(\mathrm{d}s) = \sum_{j=1}^{d} \int_{T} [M]_{s \wedge t}^{i, j}\, \nu_{j}(\mathrm{d}s) \\
& = \sum_{j=1}^{d} \int_{T} \int_{[\4 0, s \wedge t]}\pi_{i,j}(w)\, \mu(\mathrm{d}w)\, \nu_{j}(\mathrm{d}s) \\
& = \sum_{j=1}^{d} \int_{T} \pi_{i,j}(w) \int_{T}\indicatorset{[\4 0, s \wedge t]}(w)\, \nu_{j}(\mathrm{d}s)\, \mu(\mathrm{d}w) \\
& = \sum_{j=1}^{d}\int_{[\4 0,t]} \pi_{i,j}(w) \nu_{j}\bigl([w, u]\bigr)\, \mu(\mathrm{d}w) \\
& = \int_{[\4 0,t]} \pi_{i,\cdot}(w) \nu\bigl([w, u]\bigr)\, \mu(\mathrm{d}w) \\
& = \int_{[\4 0,t]} \pi_{i,\cdot}(w) J^{*}(g)(w)\, \mu(\mathrm{d}w) = (J \circ J^{*})(g)_{i}(t),
\end{align*}
where $i \in \{1,2,\hdots, d\}$ and $t \in T$.

Let us show that $J$ is injective.
Let $f_{1}, f_{2} \in \Lambda^{2}$ be such that $J(f_{1}) = J(f_{2})$, i.e. $\| J(f_{1})-J(f_{2}) \|_{\infty} = 0$, which implies in particular for $g = f_{1}-f_{2}$ that
\begin{align}\label{eq:vanish}
\int_{A}\pi(s)g(s)\, \mu(\mathrm{d}s) = 0 \in \re^{d}
\end{align}
for all $A \in \B_{T}$ of the form $A = (s, t]$ for $s < t$ in $T$.
Since the half-open intervals generate $\B_{T}$, Dynkin's theorem shows that \eqref{eq:vanish} extends to all $A \in \B_{T}$.

We will now show that $g \sim 0$, i.e. $g^{\top}\pi g = 0$\/ $\mu$-a.e.
If this were not the case, then we would have, without loss of generality, $\mu\bigl(\{ g^{\top}\pi g > 0 \}\bigr) > 0$.
We claim that $\{ g^{\top}\pi g > 0 \} \subset \{ \pi g \neq 0 \}$.
To see this, pick $s \in T$ such that $g^{\top}(s)\pi(s)g(s) > 0$, and assume that $\pi(s)g(s) = 0$.
In other words, for each $i \in \{1,2,\hdots,d\}$, we would have $\sum_{j=1}^{d}\pi_{i,j}(s)g_{j}(s) = 0$.
But this can't be the case, since we would then have
\begin{align*}
0 < g^{\top}(s)\pi(s)g(s) = \sum_{i=1}^{d}g_{i}(s) \Bigl( \sum_{j=1}^{d} \pi_{i,j}(s)g_{j}(s) \Bigr) = 0.
\end{align*}

Now $\mu\bigl(\{ g^{\top}\pi g > 0 \}\bigr) > 0$ implies that $\mu\bigl(\{ \pi g \neq 0 \}\bigr) > 0$.
Since $\{ \pi g \neq 0 \} = \bigcup_{i=1}^{d} \{ (\pi g)_{i} \neq 0 \}$, there exists $i \in \{1,2,\hdots,d\}$ such that $\mu\bigl(\{ (\pi g)_{i} \neq 0 \}\bigr) > 0$.

Without loss of generality, we may assume that $\mu\bigl(\{ (\pi g)_{i} > 0 \}\bigr) > 0$.
Note that the set $A = \{ (\pi g)_{i} > 0 \}$ can be written as
\begin{align*}
A = \bigcup_{n \in \na} \{ (\pi g)_{i} \ge \tfrac{1}{n} \} = \bigcup_{n \in \na} ((\pi g)_{i})^{-1}\bigl( [\tfrac{1}{n}, \infty) \bigr),
\end{align*}
where every $A_{n} \coloneqq ((\pi g)_{i})^{-1}\bigl( [\tfrac{1}{n}, \infty) \bigr)$ and therefore also $A$ is $\B_{T}$-measurable.
Now $\mu(A) > 0$ implies that $\mu(A_{n}) > 0$ for some $n \in \na$, hence $n\int_{A_{n}}(\pi g)_{i}(s)\, \mu(\mathrm{d}s) \ge \mu(A_{n}) > 0$, which yields a contradiction to \eqref{eq:vanish}.

We may therefore conclude that $f_{1} - f_{2} = g \sim 0$ in $\Lambda^{2}$, which shows that the operator $J \from \Lambda^{2} \to E$ is injective.
Theorem \ref{thm:factorization} now implies that the Cameron--Martin space of $\gamma_{M}$ is given by $J(\Lambda^{2}) = H$.
\end{proof}

\begin{proof}[Proof of Theorem \ref{thm:CMapprox}]
The property of $H(D)$ being a dense subset of $H$ follows from Convention \ref{convention3} and the definition of the norm on $H$ that is induced by the inner product $\langle \cdot, \cdot \rangle_{H}$.
Being a dense subset of a separable metric space implies the remaining assertion of \ref{thm:CMapprox1}.

If $D$ is also a linear subspace of $\Lambda^{2}$, then $H(D)$ is clearly an inner product space, whose completion is $H$ by \ref{thm:CMapprox1}.
We now follow a standard argument, a version of which can be found e.g. in \cite[Proposition 1]{MR840627}.
Since $H(D)$ is dense in $H$ by \ref{thm:CMapprox1} and $H$ is separable due to Proposition \ref{prop:CMspace2}, there exists a countable subset of $H(D)$ that is also dense in $H$.
Upon applying the Gram--Schmidt process to this subset, one obtains a countable set of orthonormal vectors, which are in $H(D)$, whose linear span is dense in $H$.

We know from Proposition \ref{prop:CMspace6} that $H$ is the Cameron--Martin space of $\gamma_{M}$.
By standard theory for Gaussian measures we know that the topological support of $\gamma_{M}$ then coincides with $\overbar{H}$, where the closure is taken in $E$ (see Remark \ref{app:degenerate}).
But since $H(D)$ is dense in $H$ by \ref{thm:CMapprox1}, and the canonical injection from $H$ to $E$ is continuous by \cite[Proposition 2.4.6]{MR1642391}, we have $\overbar{H} = \overbar{H(D)}$, which yields \ref{thm:CMapprox3}.
\end{proof}

\begin{proof}[Proof of Proposition \ref{prop:NNdense}]
For the purpose of the proof, we denote by $\|\4\cdot\4\|_{\infty;T}$ either the supremum or the $\lambda$-essential supremum over $T$, depending on which of the two conditions in the statement of Proposition \ref{prop:NNdense} holds.

The affine functions $\re \ni x \mapsto \alpha x + \eta$ with $\alpha, \eta \in \re$ are continuous and therefore bounded over $T$.
Since $\psi$ is (locally $\lambda$-essentially) bounded, we note that each $f \in \mathcal{NN}_{1, \infty}^{d}(\psi)$ is ($\lambda$-essentially) bounded, hence
\begin{align}\label{eq:helper}
\int_{T}f_{i}^{2}(s)\, \mu_{i,i}(\mathrm{d}s) \le \|f_{i}\|_{\infty;T}^{2} \int_{T} \pi_{i,i}(s)\, \mu(\mathrm{d}s) < \infty
\end{align}
for each $i \in \{1,2,\hdots,d\}$.
Note that in \eqref{eq:helper}, we implicitly used the fact that $\mu$ is absolutely continuous w.r.t. $\lambda$ in the case of Condition \ref{thm:UATAs3}, since in this case each $f \in \mathcal{NN}_{1, \infty}^{d}(\psi)$ is $\mu$-essentially bounded.
We conclude that $\mathcal{NN}_{1, \infty}^{d}(\psi)$ is a linear subspace of $\Lambda^{2,0}$.

For $f \in \Lambda^{2,0}$, let $\epsilon > 0$ be given.
For each $\eta \in \re^{d}$ and $s \in T$, we have $\eta^{\top}\pi(s)\eta \le |\eta|^{2} \trace(\pi(s))$ (cf. \cite[Section 3]{MR1975582}).
By Lemma \ref{lem:L2space4}, the continuous functions $C(T; \re^{d})$ are dense in $\Lambda^{2,0}$, hence there exists some $f_{\epsilon} \in C(T; \re^{d})$ such that $\| f - f_{\epsilon} \|_{\Lambda^{2}} < \epsilon / 2$.
By Theorem \ref{thm:UAT}, there exists some $g \in \mathcal{NN}_{1, \infty}^{d}(\psi)$ such that $\|f_{\epsilon} - g\|_{\infty; T} < \epsilon / (2\sqrt{\|\trace(\pi)\|_{L^{1}(\mu)}})$, hence
\begin{align*}
\| f - g \|_{\Lambda^{2}} \le \|f-f_{\epsilon}\|_{\Lambda^{2}} + \|f_{\epsilon} - g\|_{\Lambda^{2}} < \epsilon / 2 + \| f_{\epsilon} - g \|_{\infty; T} \sqrt{\|\trace(\pi)\|_{L^{1}(\mu)}} < \epsilon,
\end{align*}
which concludes our proof.
\end{proof}

\begin{proof}[Proof of Proposition \ref{prop:NNdense2}]
Since $\psi$ is bounded, one can show precisely as in the proof of Proposition \ref{prop:NNdense} that $\mathcal{NN}_{1, \infty}^{d}(\psi)$ is a linear subspace of $\Lambda^{2,0}$.
If $\mathcal{NN}_{1, \infty}^{d}(\psi)$ were not dense in $\Lambda^{2,0}$, then there would exist by the geometric version of Hahn--Banach's theorem a functional $F \in (\Lambda^{2,0})^{*}$ such that $F \neq 0$ and $F(f) = 0$ for each $f \in \mathcal{NN}_{1, \infty}^{d}(\psi)$.
Let $N$ denote the subspace of all $G \in (\Lambda^{2})^{*}$ that annihilate $\Lambda^{2,0}$, i.e. for which $G(f) = 0$ for each $f \in \Lambda^{2,0}$ holds.
$(\Lambda^{2,0})^{*}$ can then be identified with the quotient space $(\Lambda^{2})^{*} / N$.

From Lemma \ref{lem:L2space2}, we know that there exists a function $g \in \Lambda^{2}$, such that $F(f) = \int_{T} f^{\top}(s)\pi(s)g(s)\, \mu(\mathrm{d}s)$ for each $f \in \Lambda^{2,0}$.
By linearity of the Lebesgue--Stieltjes integral,
\begin{align}\label{eq:funcvanish}
F(f) = \sum_{i=1}^{d} \int_{T} f_{i}(s) \sum_{j=1}^{d} \pi_{i,j}(s) g_{j}(s)\, \mu(\mathrm{d}s) = 0,
\quad f \in \mathcal{NN}_{1, \infty}^{d}(\psi),
\end{align}
and by a variant of the Cauchy--Schwarz inequality (cf. \cite[Lemma 4.17]{MR1975582}), for each $i \in \{1,2,\hdots,d\}$ and $A \in \mathcal{B}_{T}$, it holds that
\begin{align}\label{eq:CSvariant}
\int_{A} \bigl|\sum_{j=1}^{d} \pi_{i,j}(s)g_{j}(s)\bigr|\, \mu(\mathrm{d}s) \le \sqrt{\mu_{i,i}(A)}\|g\|_{\Lambda^{2}} \le \sqrt{\mu(A)}\|g\|_{\Lambda^{2}},
\end{align}
hence $\nu_{i}(A) \coloneqq \int_{A} \sum_{j=1}^{d} \pi_{i,j}(s) g_{j}(s)\, \mu(\mathrm{d}s)$ defines a signed Borel measure on $T$ that is of finite total variation.

Precisely as in \cite{MR1015670,H1991}, we now arrive at the question whether there can exist a signed Borel measure $\nu \neq 0$ on $T$ that is of finite total variation, such that $\int_{T} \psi(\alpha x + \eta)\, \nu(\mathrm{d}x) = 0$ holds for all $\alpha, \eta \in \re$.
As we know from \cite[Lemma 1]{MR1015670}, this is not the case if $\psi$ is bounded, measurable and sigmoidal (meaning that $\psi(t) \to 0$ as $t \to -\infty$ and $\psi(t) \to 1$ as $t \to \infty$), and \cite[Theorem 5]{H1991} then generalized this finding to show that this is not the case if $\psi$ is bounded, measurable and nonconstant.
In other words, \eqref{eq:funcvanish} implies that $F \equiv 0$, which yields a contradiction.
\end{proof}

\begin{proof}[Proof of Proposition \ref{prop:asoptapprox}]
Since each $f \in D$ is a linear combination of compositions of $\psi$ and affine functions, both of which are continuously differentiable, if follows that $f$ is continuously differentiable, hence of bounded variation, which shows that $H(D)$ is a subspace of $H_{\mathrm{bv}}$.

By Proposition \ref{prop:NNdense}, there exists a sequence $(h_{n})_{n \in \na}$ in $H(D)$ that converges to $g_{h,M}$ in $H$.
From the proof of Theorem \ref{thm:CMapprox3} we know that the canonical injection from $H$ to $C_{0}(T; \re)$ is continuous, which implies that $(h_{n})_{n \in \na}$ converges to $g_{h,M}$ in $C_{0}(T; \re)$.
Since $\tilde{F} \from C_{0}(T; \re) \to \re \cup \{-\infty\}$ is assumed to be continuous, it follows that $\tilde{F}(h_{n})$ converges to $\tilde{F}(g_{h,M})$.
Moreover, since $\|\4\cdot\4\|_{H} \from H \to \replus$ is Lipschitz-continuous, we can conclude that $\tilde{F}_{h,M}(h_{n})$ converges to $\tilde{F}_{h,M}(g_{h,M})$.
\end{proof}

The following lemma follows from standard arguments.

\begin{lemma}\label{lem:IKWapprox}
For each\/ $h \in H$ and\/ $p \in [1, \infty)$, we have\/ $f_{h} \in L^{p}(M)$ and therefore\/ $f_{h}^{\top} \sbullet[.75] M \in \mathcal{H}^{p}$.
Moreover, if\/ $(h_{n})_{n \in \na}$ denotes a sequence that converges to\/ $h$ in\/ $H$, then\/ $f_{h_{n}}^{\top} \sbullet[.75] M \to f_{h}^{\top} \sbullet[.75] M$ in\/ $\mathcal{H}^{p}$ for each\/ $p \in [1, \infty)$.
\end{lemma}

\begin{proof}
Since $f_{h}$ is deterministic it is, being interpreted as a stochastic process, predictable (cf. \cite[Exercise 7.77]{Schmock2021}).
Moreover, since
\begin{align}\label{eq:LpEstim}
\|f_{h}\|_{L^{p}(M)} = \mathbb{E}\bigl[ [f_{h}^{\top} \sbullet[.75] M]_{u}^{p/2} \bigr]^{1/p} = \mathbb{E}\bigl[ \bigl((f_{h}^{\top}\pi f_{h}) \sbullet[.75] C\bigr)_{u}^{p/2} \bigr]^{1/p} = \| h \|_{H} < \infty, 
\end{align}
we have that $f_{h} \in L^{p}(M)$, and an application of Burkholder--Davis--Gundy's (BDG) inequality implies that $f_{h}^{\top} \sbullet[.75] M \in \mathcal{H}^{p}$.
Keeping in mind \eqref{eq:LpEstim}, an application of the BDG inequality then yields the existence of a positive constant $c_{p}$ such that
\begin{align*}
\| (f_{h_{n}} - f_{h})^{\top} \sbullet[.75] M \|_{\mathcal{H}^{p}} \le c_{p} \| h_{n} - h \|_{H},
\end{align*}
where the right-hand side converges to zero, as $n \to \infty$.
\end{proof}

Before we prove Lemma \ref{lem:density} let us recall for convenience a technical result, which follows e.g. from \cite[Satz 5.4]{MR2257838} or by combining the proofs of \cite[Lemma~1.34]{MR4226142} and \cite[Theorem 1.3.39]{MR3443368}.

\begin{lemma}\label{lem:Lpconv}
Fix\/ $p > 0$, and let\/ $(f_{n})_{n \in \na}$ be a sequence in\/ $L^{p}(\PP)$ such that\/ $f_{n} \to f$ in probability, where\/ $f \in L^{p}(\PP)$.
Then
\begin{align*}
\| f_{n} - f \|_{L^{p}(\PP)} \to 0 \quad \Leftrightarrow \quad \| f_{n} \|_{L^{p}(\PP)} \to \| f \|_{L^{p}(\PP)}.
\end{align*}
\end{lemma}

\begin{proof}[Proof of Lemma \ref{lem:density}]
Since both $H$ and $L^{p}(\PP)$ are metric spaces, it suffices to prove the sequential continuity of $A_{p}$.
Given $h \in H$, let $(h_{n})_{n \in \na}$ be a sequence that converges to $h$ in $H$, and set $V_{n} = (A_{p}(h_{n}))^{p}$ for $n \in \na$.
We will now show that $(V_{n})_{n \in \na}$ converges to $V = (A_{p}(h))^{p}$ in $L^{1}(\PP)$.
To this end, we apply Vitali's convergence theorem.
Let us first collect some important properties.

\begin{enumerate}
\setlength\itemsep{0.05em}
\item\label{property_1}
(convergence in probability)
By Lemma~\ref{lem:IKWapprox}, we have $f_{h_{n}}^{\top}\sbullet[.75]M \to f_{h}^{\top}\sbullet[.75]M$ in $\mathcal{H}^{2}$, hence $(f_{h_{n}}^{\top}\sbullet[.75]M)_{u} \to (f_{h}^{\top}\sbullet[.75]M)_{u}$ in $L^{2}(\PP)$ and thus convergence also holds in probability.
By the reverse triangle inequality, we have $\|h_{n}\|_{H} \to \|h\|_{H}$ as $n \to \infty$.
An application of the continuous mapping theorem implies that $\exp{(-p(f_{h_{n}}^{\top}\sbullet[.75]M)_{u} + p \|h_{n}\|_{H}^{2}/2)} = V_{n}$ converges to $\exp{(-p(f_{h}^{\top}\sbullet[.75]M)_{u} + p \|h\|_{H}^{2}/2)} = V$ in probability.

\item\label{property_2}
(boundedness in $L^{1}(\PP)$)
For each $n \in \na$,
\begin{align}\label{eq:mgale}
\begin{aligned}
\mathbb{E}[V_{n}] & = \mathbb{E}\bigl[\exp{\bigl(-p(f_{h_{n}}^{\top}\sbullet[.75]M)_{u} - p^{2}\|h_{n}\|_{H}^{2}/2\bigr)}\bigr]\exp{\bigl((p+p^{2})\|h_{n}\|_{H}^{2}/2\bigr)} \\
& = \mathbb{E}\bigl[Z_{u}^{n}\bigr]\exp{\bigl((p+p^{2})\|h_{n}\|_{H}^{2}/2\bigr)},
\end{aligned}
\end{align}
where $Z^{n} \coloneqq \mathcal{E}(-p(f_{h_{n}}^{\top}\sbullet[.75]M))$ is, by Novikov's criterion, a martinale.
As a consequence, we have $\mathbb{E}[Z_{u}^{n}] = \mathbb{E}[Z_{0}^{n}] = 1$, hence the sequence $(V_{n})_{n \in \na}$ is bounded in $L^{1}(\PP)$, since we have already established that $\|h_{n}\|_{H}$ converges to $\|h\|_{H}$, which in particular implies that $\exp{((p+p^{2})\|h_{n}\|_{H}^{2}/2)}$ converges to $\exp{((p+p^{2})\|h\|_{H}^{2}/2)}$.

\item\label{property_3}
(uniform integrability)
Fix $\varepsilon > 0$ and note that, by the same arguments that we used for Part \eqref{property_2},
\begin{align*}
\sup_{n \in \na} \mathbb{E}[ (V_{n})^{1+\varepsilon} ] < \infty.
\end{align*}
De la Vall\'{e}e Poussin's criterion implies that the set $\{V_{n} \colon n \in \na\} \subset L^{1}(\PP)$ is uniformly integrable.
\end{enumerate}

By Vitali's convergence theorem, we now have that $V_{n} \to V$ in $L^{1}(\PP)$.
If we repeat the arguments laid out in Part~\eqref{property_1} for $p = 1$, we see that $A_{p}(h_{n}) \to A_{p}(h)$ in probability.
Moreover,
\begin{align*}
\|A_{p}(h_{n})\|_{L^{p}(\PP)}^{p} = \|V_{n}\|_{L^{1}(\PP)} \to \|V\|_{L^{1}(\PP)} = \|A_{p}(h)\|_{L^{p}(\PP)}^{p},
\end{align*}
hence, Lemma \ref{lem:Lpconv} implies that $A_{p}(h_{n}) \to A_{p}(h)$ in $L^{p}(\PP)$.
Finally, note that Equation \eqref{eq:mgale} shows that $\|A_{p}(h)\|_{L^{p}(\PP)} = \mathrm{exp}{\bigl((1+p)\|h\|_{H}^{2}/2\bigr)}$, hence
\begin{align*}
\limsup_{\|h\| \to \infty} \frac{\|A_{p}\|_{L^{p}(\PP)}}{\|h\|_{H}} = \lim_{\|h\| \to \infty} \frac{\mathrm{exp}{\bigl((1+p)\|h\|_{H}^{2}/2\bigr)}}{\|h\|_{H}} = \infty,
\end{align*}
which yields the remaining assertion.
\end{proof}

\begin{proof}[Proof of Theorem \ref{thm:densityapprox}]
As in the proof of \cite[Proposition 4]{MR2680557}, we apply H{\"o}lder's inequality:
\begin{align*}
V(h) & = \mathbb{E}_{\PP}\bigl[ F^{2}(X)\exp{\bigl(-(f_{h}^{\top}\sbullet[.75]M)_{u}+\|h\|_{H}^{2}/2\bigr)}\bigr] \\
& \le \mathbb{E}_{\PP}\bigl[ |F(X)|^{2+\varepsilon}\bigr]^{1/p} \mathbb{E}_{\PP}\bigl[ \exp{\bigl(-q(f_{h}^{\top}\sbullet[.75]M)_{u}+q\|h\|_{H}^{2}/2\bigr)}\bigr]^{1/q} \\
& = \mathbb{E}_{\PP}\bigl[ |F(X)|^{2+\varepsilon}\bigr]^{1/p} \mathbb{E}_{\PP}\bigl[ \mathcal{E}\bigl(-q(f_{h}^{\top}\sbullet[.75]M)\bigr)_{u}\bigr]^{1/q} \exp{\bigl((1+q)\|h\|_{H}^{2}/2\bigr)} \\
& = \mathbb{E}_{\PP}\bigl[ |F(X)|^{2+\varepsilon}\bigr]^{1/p} \exp{\bigl((1+q)\|h\|_{H}^{2}/2\bigr)},
\quad h \in H,
\end{align*}
where $p = (2+\varepsilon)/2$ and $q = (2+\varepsilon)/\varepsilon$, and where the last equality follows from the fact that $\mathcal{E}\bigl(-q(f_{h}^{\top}\sbullet[.75]M)\bigr)$ is a martingale, hence $\mathbb{E}_{\PP}\bigl[\mathcal{E}\bigl(-q(f_{h}^{\top}\sbullet[.75]M)\bigr)_{u}\bigr] = \mathbb{E}_{\PP}\bigl[\mathcal{E}\bigl(-q(f_{h}^{\top}\sbullet[.75]M)\bigr)_{0}\bigr] = 1$.
This shows that $V$ is $\replus$-valued.

Given $h \in H$, let $(h_{n})_{n \in \na}$ be a sequence that converges to $h$ in $H$.
Due to Lemma \ref{lem:density}, $Z^{n}\coloneqq A_{q}(h_{n})$ converges to $Z \coloneqq A_{q}(h)$ in $L^{q}(\PP)$.
Note that $F^{2}(X) \in L^{p}(\PP)$ by assumption.
By Riesz's representation theorem, the topological dual of $L^{q}(\PP)$ is isometrically isomorphic to $L^{p}(\PP)$, where the isomorphism is given by
\begin{align*}
L^{p}(\PP) \ni g \mapsto \Bigl( L^{q}(\PP) \ni f \mapsto \int_{\Omega} g(\omega) f(\omega)\, \PP(\mathrm{d}\omega) \Bigr),
\end{align*}
hence the map $L^{q}(\PP) \ni Y \mapsto \mathbb{E}[F^{2}(X)Y] \in \re$ is continuous, which yields
\begin{align*}
\lim_{n \to \infty} V(h_{n}) = \lim_{n \to \infty} \mathbb{E}[ F^{2}(X)Z^{n}] = \mathbb{E}[ F^{2}(X)Z] = V(h).
\end{align*}
Since $H$ is in particular a metric space, continuity of $V$ is equivalent to sequential continuity, which shows \ref{thm:densityapprox2}.

In order to prove existence of a minimizer of $V$, we borrow some tools from convex optimization, see \cite{MR1921556} for details.
First, we show that $V$ is proper, meaning that $\{ h \in H\ |\ V(h) < \infty \} \neq \emptyset$ and $V(h) > - \infty$ for all $h \in H$.
The latter condition is clearly satisfied, as $V$ is nonnegative.
For $h \equiv 0$, we further have $V(h) = \mathbb{E}[ F^{2}(X)] < \infty$, which implies the former condition (which also follows from \ref{thm:densityapprox1}).
Moreover, since $V$ is continuous as argued above, it is in particular lower semicontinuous.

Let us show that $V$ is coercive, i.e. that $V(h) \to \infty$, as $\|h\|_{H} \to \infty$.
Since we assume that $\PP\bigl(\{ F^{2}(X) > 0 \} \bigr) > 0$, there exists a constant $\delta > 0$ such that $\PP\bigl(\{ F^{2}(X) \ge \delta \} \bigr) > 0$.
An application of the reverse H{\"o}lder inequality along the lines of the proof of \cite[Proposition 4]{MR2680557} reveals the inequality
\begin{align}\label{eq:coercive}
V(h) \ge \delta\, \PP\bigl(\{ F^{2}(X) \ge \delta \} \bigr)^{3} \exp{\bigl(\|h\|_{H}^{2}/4\bigr)},
\quad h \in H,
\end{align}
which shows that $V$ is coercive.

Finally, we need to show that $V$ is convex.
To this end, pick $\eta \in (0,1)$ and $g,h \in H$ such that $g \neq h$.
By the triangle inequality and positive homogeneity, we have $\| \eta g + (1-\eta)h \|_{H} \le \eta \| g \|_{H} + (1-\eta) \| h \|_{H}$.
By the convexity of $\re \ni x \mapsto x^{2}$ and linearity of the vector stochastic integral, we thus have
\begin{align*}
- \bigl( (\eta f_{g} & + (1-\eta) f_{h})^{\top} \sbullet[.75] M \bigr)_{u} + \| \eta g + (1-\eta) h \|_{H}^{2} / 2 \\
& \le \eta \bigl( - (f_{g}^{\top} \sbullet[.75] M)_{u} + \| g \|_{H}^{2} / 2 \bigr) +  (1-\eta) \bigl( - (f_{h}^{\top} \sbullet[.75] M)_{u} + \| h \|_{H}^{2} / 2 \bigr).
\end{align*}
Together with the convexity and monotonicity of $\re \ni x \mapsto \exp{x}$, this shows that $V$ is convex.

Consequently, \cite[Proposition 2.5.6]{MR1921556} shows that $\argmin_{h \in H}V(h)$ is a convex set.
Moreover, upon noting that $H$, being a Hilbert space, is reflexive, \cite[Theorem 2.5.1]{MR1921556} shows that $\argmin_{h \in H}V(h)$ is not empty, which shows \ref{thm:densityapprox3}.

For $\delta > 0$ choose $h_{\delta} \in H$ such that $V(h_{\delta}) < \min_{h \in H}V(h) + \delta$.
By Theorem \ref{thm:CMapprox1}, there exists a sequence $(h_{n})_{n \in \na}$ in $H(D)$ that converges to $h_{\delta}$ in $H$.
By \ref{thm:densityapprox2} we then obtain $\lim_{n \to \infty} V(h_{n}) = V(h_{\delta}) < \min_{h \in H}V(h) + \delta$.
A diagonalization argument yields \ref{thm:densityapprox4}.
\end{proof}

\bibliographystyle{amsplain}
\bibliography{Importance_sampling_with_feedforward_neural_networks}

\end{document}